\documentclass[11pt,letterpaper]{article}
\usepackage{fullpage}
\usepackage[top=2cm, bottom=4.5cm, left=2.5cm, right=2.5cm]{geometry}
\usepackage{amsmath,amsthm,amsfonts,amssymb,amscd}
\usepackage{lastpage}
\usepackage{enumerate}
\usepackage{fancyhdr}
\usepackage{mathrsfs}
\usepackage{xcolor}
\usepackage{graphicx}
\usepackage{listings}
\usepackage{framed}
\usepackage[colorlinks=true, urlcolor=blue, linkcolor=blue, citecolor=magenta]{hyperref}
\usepackage{cite}
\usepackage{multirow}
\usepackage{comment}
\usepackage{footnote}
\usepackage{tikz}

\hypersetup{%
  colorlinks=true,
  linkcolor=red,
  linkbordercolor={0 0 1}
}

\lstdefinestyle{Matlab}{
    language        = matlab,
    frame           = lines, 
    basicstyle      = \footnotesize,
    keywordstyle    = \color{blue},
    stringstyle     = \color{green},
    commentstyle    = \color{red}\ttfamily
}
\newtheorem{theorem}{Theorem}
\newtheorem{definition}[theorem]{Definition}

\newtheorem{remark}[theorem]{Remark}
\newtheorem{lemma}[theorem]{Lemma}
\newtheorem{hypothesis}[theorem]{Hypothesis}
\newtheorem{corollary}[theorem]{Corollary}

\newtheorem{openproblem}[theorem]{Open Problem}

\allowdisplaybreaks[4]

\newcommand{\wtwo}{\mathsf{W[2]}}
\newcommand{\wone}{\mathsf{W[1]}}
\newcommand{\fpt}{\mathsf{FPT}}

\newcommand{\F}{\mathbb{F}}
\newcommand{\N}{\mathbb{N}}
\newcommand{\WOne}{\textsf{W[1]}}

\newcommand{\FPT}{\textsf{FPT}}
\newcommand{\NP}{\textsf{NP}}

\newcommand{\nc}{\newcommand}

\newlength{\probwidth}
\nc{\prob}[3][9]{
\begin{center}
  \normalfont\fbox{
   \begin{tabular}[t]{
     rp{#1cm}}\textit{Instance:}&#2. \\
     \textit{Problem:}&#3
   \end{tabular}}
\end{center}}

\nc{\pprob}[4][9]{
\begin{center}
   \normalfont\fbox{
    \begin{tabular}[t]{
     rp{#1cm}}\textit{Instance:}&#2. \\
     \textit{Parameter:}&#3. \\
     \textit{Problem:}&#4
   \end{tabular}}
\end{center}}

\nc{\nprob}[4][9]{
\begin{center}
  \normalfont\fbox{

\addtolength{\probwidth}{#1cm}\parbox{\probwidth}{\textsc{#2}\\\hspace*{1.5em}
     \begin{tabular}[t]{
      rp{#1cm}}\textit{Instance:}&#3. \\
      \textit{Problem:}&#4
     \end{tabular}}}
\end{center}}

\nc{\npprob}[5][9]{
\begin{center}
  \normalfont\fbox{

\addtolength{\probwidth}{#1cm}\parbox{\probwidth}{\textsc{#2}\\\hspace*{1.5em}
    \begin{tabular}[t]{
     rp{#1cm}}\textit{Instance:}&#3. \\
     \textit{Parameter:}&#4. \\
     \textit{Problem:}&#5
    \end{tabular}}}
\end{center}}

\nc{\nppxrob}[5][9]{ \normalfont\fbox{

\addtolength{\probwidth}{#1cm}\parbox{\probwidth}{\textsc{#2}\\\hspace*{1.5em}
   \begin{tabular}[t]{
    rp{#1cm}}\textit{Instance:}&#3. \\
    \textit{Parameter:}&#4. \\
    \textit{Problem:}&#5
   \end{tabular}}}}

\nc{\nppprob}[5][4]{
\begin{center}
  \normalfont\fbox{

\addtolength{\probwidth}{#1cm}\parbox{\probwidth}{\textsc{#2}\\\hspace*{1.5em}
    \begin{tabular}[t]{
     rp{#1cm}}\textit{Instance:}&#3. \\
     \textit{Parameter:}&#4. \\
     \textit{Problem:}&#5
    \end{tabular}}}
\end{center}}

\nc{\noptprob}[6][9]{
\begin{center}
  \normalfont\fbox{

\addtolength{\probwidth}{#1cm}\parbox{\probwidth}{\textsc{#2}\\\hspace*{1.5em}
    \begin{tabular}[t]{
     rp{#1cm}}\textit{Instance:}&#3. \\
     \textit{Solution:}&#4. \\
     \textit{Cost:}&#5. \\
     \textit{Goal:}&#6.
    \end{tabular}}}
\end{center}}

\nc{\npprobyn}[6][9]{
\begin{center}
  \normalfont\fbox{

\addtolength{\probwidth}{#1cm}\parbox{\probwidth}{\textsc{#2}\\\hspace*{1.5em}
    \begin{tabular}[t]{
     rp{#1cm}}\textit{Instance:}&#3. \\
     \textit{Parameter:}&#4. \\
     \textit{Problem:}&Distinguish between the following two cases: \\
     \textbf{(YES)}&#5. \\
     \textbf{(NO)}&#6.
    \end{tabular}}}
\end{center}}

\begin{document}

\title{Improved Lower Bounds for Approximating Parameterized Nearest Codeword  and Related Problems under ETH}

\author{
Shuangle Li\thanks{Nanjing University. Email: \texttt{shuangleli@smail.nju.edu.cn}}
\and
Bingkai Lin\thanks{Nanjing University. Email: \texttt{lin@nju.edu.cn}}
\and
Yuwei Liu\thanks{Shanghai Jiao Tong University. Email: \texttt{yuwei.liu@sjtu.edu.cn}}
}

\maketitle

\begin{abstract}
In this paper we present a new gap-creating randomized self-reduction for parameterized Maximum Likelihood Decoding problem over $\mathbb{F}_p$ (\textsc{$k$-MLD$_p$}). 
The reduction takes a \textsc{$k$-MLD$_p$} instance with $k\cdot n$ vectors as input, runs in time $f(k)n^{O(1)}$ for some computable function $f$, outputs a \textsc{$(3/2-\varepsilon)$-Gap-$k'$-MLD$_p$} instance for any $\varepsilon>0$, where $k'=O(k^2\log k)$.
Using this reduction, we show that assuming the randomized Exponential Time Hypothesis (ETH), no algorithms can approximate \textsc{$k$-MLD$_p$} (and therefore its dual problem \textsc{$k$-NCP$_p$}) within factor $(3/2-\varepsilon)$ in $f(k)\cdot n^{o(\sqrt{k/\log k})}$ time for any $\varepsilon>0$.

We then use reduction by Bhattacharyya, Ghoshal, Karthik and Manurangsi (ICALP 2018) to amplify the $(3/2-\varepsilon)$-gap to any constant. As a result, we show that assuming ETH, no algorithms can approximate \textsc{$k$-NCP$_p$} and \textsc{$k$-MDP$_p$} within $\gamma$-factor in $f(k)n^{o(k^{\varepsilon_\gamma})}$ time for some constant $\varepsilon_\gamma>0$. Combining with the gap-preserving reduction by Bennett, Cheraghchi, Guruswami and Ribeiro (STOC 2023), we also obtain similar lower bounds for \textsc{$k$-MDP$_p$}, \textsc{$k$-CVP$_p$} and \textsc{$k$-SVP$_p$}.

These results improve upon the previous $f(k)n^{\Omega(\mathsf{poly} \log k)}$ lower bounds for these problems under ETH using reductions by Bhattacharyya et al. (J.ACM 2021) and Bennett et al. (STOC 2023).
\end{abstract}

\section{Introduction}
The study of linear error correcting codes has drawn attention to two dual fundamental computational problems called \textsc{Nearest Codeword Problem (NCP)} and \textsc{Maximum Likelihood Decoding (MLD)}.
Given a matrix $A\in\mathbb{F}_p^{m\times n}$ and a vector $\vec{t}\in\mathbb{F}_p^m$, the \textsc{Nearest Codeword Problem (NCP)} asks for a vector $\vec{x}\in\mathbb{F}_p^n$ such that $||A\vec x-\vec t||_0$ is minimized. Here $||\cdot||_0$ denotes the Hamming weight. While in the \textsc{Maximum Likelihood Decoding (MLD)}, we are given 
a matrix $A\in\mathbb{F}_p^{m\times n}$ and a vector $\vec{t}\in\mathbb{F}_p^m$, the goal is to minimize $||\vec{x}||_0$ subject to $A\vec x=\vec t$.
Another fundamental problem related to a linear code is the homogeneous version of \textsc{NCP}, known as \textsc{Minimum Distance Problem (MDP)}, where the task is to find a non-zero vector $\vec x$ such that $||A\vec x||_0$ is minimized. 

The computational complexity of \textsc{MLD}, \textsc{NCP} and \textsc{MDP} has been studied with great effort throughout the past several decades.
It is known that \textsc{MLD}, \textsc{NCP} and \textsc{MDP} are not only \NP-hard~\cite{BMT78,Vardy97}, but also \NP-hard to approximate within any constant ratio~\cite{Ste93,ABSS97,DMS03,CW12,AK14,Mic14}. Moreover, the variant of \textsc{MLD} that allows the code being preprocessed by unbounded computational resource is also \NP-hard to approximate within a factor of $(3-\varepsilon)$\cite{FM04,Regev04}.
Also it is proven that assuming $\NP\not\subseteq\mathsf{DTIME}(n^{\mathsf{poly}(\log n)})$, no polynomial time algorithm can approximate \textsc{NCP} up to $2^{\log^{1-\epsilon}n}$ factor for any $\epsilon>0$~\cite{ABSS97,Raz98} and no polynomial time algorithm can approximate \textsc{MDP} up to $2^{\log^{1-\epsilon}n}$ for any $\epsilon>0$~\cite{DMS03,CW12,AK14,Mic14}.
For some specific codes, \textsc{MLD} is also shown to be \NP-hard, e.g. product code\cite{Barg1994}, Reed-Solomon code\cite{GV05}, algebraic geometry code\cite{Cheng08a}.
On the algorithmic side, it is known that NCP can be approximate to $O(n/\log n)$ in polynomial time~\cite{AlonPY09}.

The lattice version of \textsc{NCP} and \textsc{MDP} are known as \textsc{Closest Vector Problem (CVP)} and \textsc{Shortest Vector Problem (SVP)}.
In these problems, a lattice $\mathcal L$ is given instead of a linear code.
For \textsc{CVP} a target $\vec t$ is additionally given and the goal is to find a vector $\vec v\in \mathcal L$ such that $||\vec v-\vec t||_p$ is minimized, where $||\cdot||_p$ denotes the $\ell_p$-norm. 
And for \textsc{SVP} the goal is to find a non-zero vector $\vec v\in\mathcal{L}$ with minimum $\ell_p$ norm.
The study for \textsc{CVP} and \textsc{SVP} also has long history~\cite{Ste93,ABSS97,Ajtai98,GMSS99,Mic00,Mic01,DKRS03,Khot05,HR12,Mic14}.
For \textsc{CVP}, it is \NP-hard to approximate within factor $n^{c/\log\log n}$ for some constant $c>0$~\cite{DKRS03}.
As for \textsc{SVP}, it was shown that no polynomial time algorithm can approximate \textsc{SVP} within any constant factor assuming $\NP\not\subseteq\mathsf{RP}$~\cite{Khot05}, and no polynomial time algorithm can approximate \textsc{SVP} up to $2^{\log ^{1-\epsilon} n}$ factor assuming $\NP\not\subseteq\mathsf{RTIME}(n^{\mathsf{poly}(\log n)})$~\cite{HR12}. Lattice problems have many applications in  cryptography~\cite{Reg09JACM, Reg10}. 
Due to their importance, lattice problems are also extensively studied in the fine-grained complexity area, see, e.g.,~\cite{AS18,ABGS21,BPT22,ABBG0LPS23} and a very recent survey by Bennett~\cite{Bennett23} for more details on hardness of \textsc{SVP}.

Over the past three decades, parameterized complexity, a new framework to address \NP-hard problems, has been rapidly developed and drawing growing attention.
The study in the field of parameterized complexity focuses on whether a problem can be solved in $f(k)\cdot n^{O(1)}$ time (\FPT~time), where $k$ is a parameter given along with the instance.
In the parameterized version of \textsc{$k$-MLD}, \textsc{$k$-NCP}, \textsc{$k$-MDP}, \textsc{$k$-CVP} and \textsc{$k$-SVP}, an integer $k$ is additionally given and the task is to decide whether the optimal value is no greater than $k$.
Downey, Fellows, Vardy and Whittle \cite{DFVW99} showed that \textsc{$k$-MLD} (and therefore \textsc{$k$-NCP}) is $\wone$-hard and belongs to $\wtwo$. They asked if \textsc{$k$-CVP} and \textsc{$k$-SVP} (in $\ell_2$ norm) is $\wone$-hard.
20 years later in recent breakthroughs~\cite{BBE+21,BCGR23}, the parameterized intractability of \textsc{$k$-NCP}, \textsc{$k$-MDP}, \textsc{$k$-CVP} and \textsc{$k$-SVP} are  settled.
Notably they ruled out not only exact \FPT~algorithms, but also \FPT~approximation algorithms as well.
Specifically, 
\cite{BBE+21} first presented a gap-creating reduction for \textsc{$k$-NCP} and then showed gap-preserving reductions from \textsc{$k$-NCP} towards \textsc{$k$-MDP}, \textsc{$k$-CVP} and \textsc{$k$-SVP}. Soon afterwards, Bennett, Cheraghchi, Guruswami and Ribeiro~\cite{BCGR23} improved the gap-preserving reductions for more general cases (general fields and general $\ell_p$ norm). These two works jointly showed that it is $\wone$-hard to approximate \textsc{$k$-NCP} and \textsc{$k$-MDP} within any constant factor over any finite field $\F_p$, and it is $\wone$-hard to approximate \textsc{$k$-CVP} in the $\ell_p$ norm within any constant factor for any $p\geq 1$.
And they showed hardness for \textsc{$k$-SVP} to approximate within any constant factor in the $\ell_p$ norm for any $p>1$, and some constant approaching $2$ for $p=1$ .

After obtaining \FPT-inapproximability results, it is natural to study fine-grained time lower bounds for parameterized approximability of these problems. 
Assuming Gap-ETH \cite{dinur2016mildly,manurangsi2016birthday}, Manurangsi~\cite{Man20} showed that no $f(k)\cdot n^{o(k)}$ time algorithm can approximate \textsc{$k$-NCP} and \textsc{$k$-CVP} to any constant factor.  With  the gap-preserving reduction in~\cite{BCGR23}, one can further show that no $f(k)\cdot n^{o(k)}$ time algorithm can approximate \textsc{$k$-MDP} and \textsc{$k$-SVP} to any constant under the randomized Gap-ETH. All these results are based on an assumption with a gap. This raises the following open question:
\begin{itemize}
    \item[(1)] Can we establish similar lower bounds for these problems under the weaker and gap-free assumption 
of ETH?
\end{itemize}
We note that the gap-preserving reduction in~\cite{BCGR23} from \textsc{Gap-$k$-NCP} (\textsc{Gap-$k$-CVP}) to \textsc{Gap-$k'$-MDP} (\textsc{Gap-$k'$-SVP}) has $k'=O(k)$. So, it suffices to prove constant \textsc{Gap-$k$-NCP} (\textsc{Gap-$k$-CVP}) has no $f(k)\cdot n^{o(k)}$-time algorithm assuming ETH \cite{IP01}. Unfortunately, the gap-creating reduction in~\cite{BBE+21} causes an exponential growth of the parameter and only gives an $\Omega(n^{(\log k)^{1/(2+\epsilon)}})$-time lower bound for constant \textsc{Gap-$k$-NCP} under ETH (See the analysis in Section~\ref{sec:prework}). Therefore, finding better reductions for \textsc{Gap-$k$-NCP} and \textsc{Gap-$k$-CVP} is the crux of improving  lower bounds for \textsc{Gap-$k$-MDP} and \textsc{Gap-$k$-SVP}.

\subsection{Our Contributions}
We take a step forward on closing the gap between results under gap-free assumption (ETH) and gap assumption (Gap-ETH). Our main result is a new direct gap-creating self reduction for \textsc{$k$-MLD}, which is the dual problem of \textsc{$k$-NCP}, with polynomial growth of the parameter.
\begin{theorem}[informal; See Theorem~\ref{theorem: gap creating (main)} for a formal statement]
    For any constant $1<\gamma<\frac{3}{2}$ and prime power $p>1$, there is a reduction runs in $O_k(n^{O(1)})$ that on input a \textsc{$k\text{-MLD}_p$} instance $(V, \vec{t})$, output a \textsc{Gap-$k$-MLD$_p$} instance $(V', \vec{t}')$ satisfies:
    \begin{itemize}
        \item (Completeness) If there exists $k$ vectors in $V$ with their sum \footnote{The definition of \textsc{$k$-MLD} used in our proof is a slightly different variant, where the vectors directly sum up to the target in the YES case, but they are essentially equivalent, see Section~\ref{prelim:probs} for more details.} being $\vec{t}$, then there exists $k'$ vectors in $V'$ with their sum being $\vec{t}'$.
        \item (Soundness) If for any set $S\subseteq V$ with size at most $k$, $\vec{t}\notin\text{Span}(S)$, then for any set $S'\subseteq V'$ with size at most $\gamma k'$, $\vec{t}'\notin\text{Span}(S')$.
        \item Polynomial parameter growth $k'=O(k^2\log k)$. (And $k'=O(k^3)$ if not allowing randomness).
    \end{itemize}
\end{theorem}

Combining this gap-creating reduction with the $f(k)n^{\Omega(k)}$-time ETH lower bound for \textsc{$k$-MLD} in~\cite[Theorem 11]{LinRSW22}, we obtain improved lower bounds for \textsc{Gap-$k$-NCP} assuming ETH and randomized ETH.

\begin{corollary}
    Assuming randomized ETH, for any prime power $p>1$ and real number $\gamma\in(1,\frac{3}{2})$, no $f(k)n^{o(\sqrt{k/\log k})}$ time algorithm can solve \textsc{$\gamma$-Gap-$k$-NCP$_p$}.
\end{corollary}

\begin{corollary}
    Assuming ETH, for any prime power $p>1$ and real number $\gamma\in(1,\frac{3}{2})$, no $f(k)n^{o(k^{1/3})}$ time algorithm can solve \textsc{$\gamma$-Gap-$k$-NCP$_p$}.
\end{corollary}

By applying the gap amplification procefure in~\cite{BGKM18} ($\gamma \rightarrow \Omega(\gamma^2), k\rightarrow O(k^2)$, see Theorem \ref{theorem: gap amplification for mld in BGKM18} for a formal statement) sufficiently many (but still constant) times, we obtain a reduction for \textsc{Gap-$k$-MLD} with any constant gap with still polynomial growth of parameter. Therefore we obtain the following improved ETH lower bound for \textsc{$k$-NCP}.

\begin{corollary}
    Assuming ETH, for any prime power $p>1$ and real number $\gamma>1$, no $f(k)n^{o(k^\epsilon)}$ time algorithm can solve \textsc{$\gamma$-Gap-$k$-NCP$_p$} where $\epsilon=\frac{1}{\textsf{polylog}(\gamma)}$ is a constant.
\end{corollary}

Combining our results of \textsc{Gap-$k$-NCP$_p$} with the gap-preserving reductions in \cite{BBE+21} and \cite{BCGR23}, we obtain improved ETH lower bounds for constant approximating \textsc{$k$-NCP}, \textsc{$k$-CVP}, \textsc{$k$-MDP} and \textsc{$k$-SVP}. The summarize of corollaries are present in Table~\ref{table: corollaries}.

\begin{savenotes}
\begin{table}[h]
\begin{center}
\begin{tabular}{c|c|c|c|c}
    \multicolumn{5}{c}{\bf Summarize of Corollaries}\\
    \hline
    {\bf \normalsize{Problem}}
    & {\bf \normalsize{Inapprox Factor}}
    & {\bf \normalsize{Lower Bound}}
    & {\bf \normalsize{Dependency}}
    & {\bf \normalsize{Specification}}\\
    \hline
    \textsc{$k$-NCP} & any $\gamma\in(1,\frac{3}{2})$ & $f(k)n^{\Omega(\sqrt{k/\log k})}$ &  & any finite field $\F_p$\\
    \textsc{$k$-NCP} & any $\gamma>1$ & $f(k)n^{\Omega(k^\epsilon)}$ & $\epsilon=\frac{1}{\textsf{polylog}(\gamma)}$ & any finite field $\F_p$\\
    \textsc{$k$-MDP} & any $\gamma>1$ & $f(k)n^{\Omega(k^\epsilon)}$ & $\epsilon=\frac{1}{p\log \gamma\cdot\textsf{polylog}(p)}$ & any finite field $\F_p$\\
    \textsc{$k$-CVP} & any $\gamma>1$ & $f(k)n^{\Omega(k^\epsilon)}$ & $\epsilon=\Theta(\frac{1}{\mathsf{polylog}(\gamma)})$ & in any $\ell_p$ norm, $p\geq 1$\\
    \textsc{$k$-SVP} & any $\gamma>1$ & $f(k)n^{\Omega(k^\epsilon)}$ & $\epsilon=\epsilon(p,\gamma)$\footnote{The constant $\epsilon$ is rather complicated and has no closed form, see Theorem~\ref{thm:CVP to SVP}.
    } & in any $\ell_p$ norm, $p>1$\\
    \textsc{$k$-SVP} & any $\gamma\in[1,2)$ & $f(k)n^{\Omega(k^\epsilon)}$ & $\epsilon=\epsilon(p,\gamma)$\footnote{Same reason as above, see Theorem~\ref{thm:NCP to SVP}.} & in any $\ell_p$ norm, $p\geq 1$\\
    \hline
\end{tabular}
\end{center}
\caption{The $f(k)n^{\Omega(k^\epsilon)}$-time lower bound for \textsc{$k$-NCP} and \textsc{$k$-CVP} are based on ETH. The other lower bounds are based on randomized ETH.}
\label{table: corollaries}
\end{table}
\end{savenotes}

\subsection{Technical Overview of Gap Creation Step}
\label{Introduction-Technical Overview-Gap Creation}
We implicitly use the  threshold graph composition method~\cite{Lin18,Lin19,BKN21,LRSW23} to  construct a $(3/2-\varepsilon)$-gap producing reduction for the $k$-MLD problem.   
This technique was first introduced  in \cite{Lin18} to prove the $\wone$-hardness of \textsc{$k$-Biclique} problem. A threshold graph is a bipartite graph that has a ``threshold property'', meaning that there is a significant gap in the number of common neighbors between any $k$ vertices and any $k+1$ vertices on the left side. Threshold graph and its variants have been widely used to show hardness of approximation for various parameterized problems, such as \textsc{$k$-DominatingSet} \cite{CL19}, \textsc{$k$-SetCover} \cite{Lin19,KN21}, \textsc{$k$-SetIntersection} \cite{BKN21} or to create gap for subsequent reductions, e.g. \cite{BBE+21}.

Let $\dot\cup$ denotes for union set of multiple disjoint sets. In this paper, we implicitly use the strong threshold graphs in  \cite{LRSW23}, which are bipartite graphs $T=(A\dot\cup B,E_T)$ with the following properties:
\begin{description}
    \item[(i)] $A=A_1\dot\cup A_2\dot\cup\cdots\dot\cup A_k$.
    \item[(ii)] $B=B_1\dot\cup B_2\dot\cup\cdots\dot\cup B_m$.
    \item[(iii)] For any $a_1\in A_1,\ldots,a_k\in A_k$ and $i\in [m]$, $a_1,\ldots,a_k$ have a common neighbor in $B_i$.
    \item[(iv)] For any $X\subseteq A$ and $I\subseteq [m]$ with $|I|\ge\varepsilon m$, if for every $i\in I$, there exists $b_i\in B_i$ has $k+1$ neighbors in $X$, then $|X|>h$.
\end{description}
These strong threshold graphs are constructed from error-correcting codes with large relative distance ($1-o(1)$), and such ``threshold'' properties essentially come from the following intuition of ECC: \emph{If there is a collection of codewords $(X)$, and a constant fraction of entries of these codewords $(I\subseteq[m],|I|\geq \varepsilon m)$ such that, for each entry $(i\in I)$, there exists two distinct codewords in the collection that having same content in it. Then, the collection must have huge size $($at least $h)$.} To characterize the aforementioned property, previous works~\cite{KN21,LRSW23} introduced the definition of ($\varepsilon$-)\emph{Collision Number} of an error-correcting code $C$, $\mathrm{Col}_{\varepsilon}(C)$, which is the minimum size of $X$ mentioned above.

\noindent\textbf{Diving into coding-based threshold graph.} Our construction deeply relies on the collision number of an ECC, so we only use threshold graph as an intuitive illustration for readers, and we directly use the error-correcting codes in our formal analysis. An informal but intuitive pictorial illustration of our construction is in Figure~\ref{Pictorial illustration in introduction}.

\begin{figure}[h]
\centering
\scalebox{0.8}{
\tikzset{every picture/.style={line width=0.75pt}} 

\begin{tikzpicture}[x=0.75pt,y=0.75pt,yscale=-1,xscale=1]

\draw (21,13) node [anchor=north west][inner sep=0.75pt]   [align=left] {For $\displaystyle i\in [ k]$, $\displaystyle \vec{v} \in V_{i}$ associated with codeword $\displaystyle C(\vec{v}) \in \Sigma ^{m}$, \\corresponds to the following vector in $\displaystyle A_{i}$:};
\draw (195,58) node [anchor=north west][inner sep=0.75pt]   [align=left] {$\displaystyle i$-th part of each segment};
\draw (454,60) node [anchor=north west][inner sep=0.75pt]   [align=left] {$\displaystyle i$-th entry};
\draw (110,160) node [anchor=north west][inner sep=0.75pt]   [align=left] {$\displaystyle m$ segments, one-hot encoding of each entry in $\displaystyle C(\vec{v})$};
\draw   (20,110) -- (590,110) -- (590,130) -- (20,130) -- cycle ;
\draw [line width=3]    (60,110) -- (60,130) ;
\draw [line width=3]    (160,110) -- (160,130) ;
\draw [line width=3]    (260,110) -- (260,130) ;
\draw    (100,110) -- (100,130) ;
\fill[color={rgb, 255:red, 30; green, 255; blue, 0 }]   (101,111) -- (119,111) -- (119,129) -- (101,129) -- cycle ;
\draw    (120,110) -- (120,130) ;
\draw    (200,110) -- (200,130) ;
\fill[color={rgb, 255:red, 255; green, 246; blue, 100 }]   (201,111) -- (219,111) -- (219,129) -- (201,129) -- cycle ;
\draw    (220,110) -- (220,130) ;
\draw [line width=3]    (350,110) -- (350,130) ;
\draw [line width=3]    (450,110) -- (450,130) ;
\draw    (390,110) -- (390,130) ;
\fill[color={rgb, 255:red, 0; green, 255; blue, 246 }]   (391,111) -- (409,111) -- (409,129) -- (391,129) -- cycle ;
\draw    (410,110) -- (410,130) ;
\draw    (520,110) -- (520,130) ;
\draw    (480,110) -- (480,130) ;
\draw    (490,110) -- (490,130) ;
\draw    (190,80) -- (111.87,107.3) ;
\draw [shift={(110,108)}, rotate = 339.44] [color={rgb, 255:red, 0; green, 0; blue, 0 }  ][line width=0.75]    (10.93,-3.29) .. controls (6.95,-1.4) and (3.31,-0.3) .. (0,0) .. controls (3.31,0.3) and (6.95,1.4) .. (10.93,3.29)   ;
\draw    (240,80) -- (211.41,106.59) ;
\draw [shift={(210,108)}, rotate = 315] [color={rgb, 255:red, 0; green, 0; blue, 0 }  ][line width=0.75]    (10.93,-3.29) .. controls (6.95,-1.4) and (3.31,-0.3) .. (0,0) .. controls (3.31,0.3) and (6.95,1.4) .. (10.93,3.29)   ;
\draw    (360,80) -- (398.4,106.8) ;
\draw [shift={(400,108)}, rotate = 216.87] [color={rgb, 255:red, 0; green, 0; blue, 0 }  ][line width=0.75]    (10.93,-3.29) .. controls (6.95,-1.4) and (3.31,-0.3) .. (0,0) .. controls (3.31,0.3) and (6.95,1.4) .. (10.93,3.29)   ;
\draw    (485,80) -- (485,106) ;
\draw [shift={(485,108)}, rotate = 270] [color={rgb, 255:red, 0; green, 0; blue, 0 }  ][line width=0.75]    (10.93,-3.29) .. controls (6.95,-1.4) and (3.31,-0.3) .. (0,0) .. controls (3.31,0.3) and (6.95,1.4) .. (10.93,3.29)   ;
\draw    (170,160) -- (111.79,132.89) ;
\draw [shift={(110,132)}, rotate = 26.57] [color={rgb, 255:red, 0; green, 0; blue, 0 }  ][line width=0.75]    (10.93,-3.29) .. controls (6.95,-1.4) and (3.31,-0.3) .. (0,0) .. controls (3.31,0.3) and (6.95,1.4) .. (10.93,3.29)   ;
\draw    (230,160) -- (211.11,133.66) ;
\draw [shift={(210,132)}, rotate = 56.31] [color={rgb, 255:red, 0; green, 0; blue, 0 }  ][line width=0.75]    (10.93,-3.29) .. controls (6.95,-1.4) and (3.31,-0.3) .. (0,0) .. controls (3.31,0.3) and (6.95,1.4) .. (10.93,3.29)   ;
\draw    (390,160) -- (399.37,133.9) ;
\draw [shift={(400,132)}, rotate = 108.43] [color={rgb, 255:red, 0; green, 0; blue, 0 }  ][line width=0.75]    (10.93,-3.29) .. controls (6.95,-1.4) and (3.31,-0.3) .. (0,0) .. controls (3.31,0.3) and (6.95,1.4) .. (10.93,3.29)   ;
\draw (34,112.4) node [anchor=north west][inner sep=0.75pt]  [font=\large]  {$\vec{v}$};
\draw (77,114.4) node [anchor=north west][inner sep=0.75pt]  [font=\normalsize]  {$0$};
\draw (131,114.4) node [anchor=north west][inner sep=0.75pt]  [font=\normalsize]  {$0$};
\draw (177,114.4) node [anchor=north west][inner sep=0.75pt]  [font=\normalsize]  {$0$};
\draw (237,114.4) node [anchor=north west][inner sep=0.75pt]  [font=\normalsize]  {$0$};
\draw (367,114.4) node [anchor=north west][inner sep=0.75pt]  [font=\normalsize]  {$0$};
\draw (421,114.4) node [anchor=north west][inner sep=0.75pt]  [font=\normalsize]  {$0$};
\draw (461,114.4) node [anchor=north west][inner sep=0.75pt]  [font=\normalsize]  {$0$};
\draw (497,114.4) node [anchor=north west][inner sep=0.75pt]  [font=\normalsize]  {$0$};
\draw (547,114.4) node [anchor=north west][inner sep=0.75pt]  [font=\normalsize]  {$0$};
\draw (480,114.4) node [anchor=north west][inner sep=0.75pt]  [font=\normalsize]  {$1$};
\draw (295,116) node [anchor=north west][inner sep=0.75pt]  [font=\normalsize]  {$\cdots $};

\draw (34,277.4) node [anchor=north west][inner sep=0.75pt]  [font=\normalsize]  {$0$};
\draw (107,277.4) node [anchor=north west][inner sep=0.75pt]  [font=\normalsize]  {$0$};
\draw (397,277.4) node [anchor=north west][inner sep=0.75pt]  [font=\normalsize]  {$0$};
\draw (481,277.4) node [anchor=north west][inner sep=0.75pt]  [font=\normalsize]  {$0$};
\draw (531,277.4) node [anchor=north west][inner sep=0.75pt]  [font=\normalsize]  {$0$};
\draw (550,277.4) node [anchor=north west][inner sep=0.75pt]  [font=\normalsize]  {$1$};
\draw (221,277.4) node [anchor=north west][inner sep=0.75pt]  [font=\normalsize]  {$\cdots $};
\draw (295,279) node [anchor=north west][inner sep=0.75pt]  [font=\normalsize]  {$\cdots $};
\draw (567,277.4) node [anchor=north west][inner sep=0.75pt]  [font=\normalsize]  {$0$};
\draw (112,225) node [anchor=north west][inner sep=0.75pt]   [align=left] {only the $\displaystyle j$-th segment is non-zero};
\draw (527,225) node [anchor=north west][inner sep=0.75pt]   [align=left] {$\displaystyle j$-th entry};
\draw (85,317) node [anchor=north west][inner sep=0.75pt]   [align=left] {(minus) one-hot encoding of each entry in $\displaystyle \vec{b}$};
\draw (21,200) node [anchor=north west][inner sep=0.75pt]   [align=left] {For $\displaystyle j\in [ m]$, vector in $\displaystyle B_{j}$ that corresponds to $\displaystyle \vec{b} \in \Sigma ^{k}$:};
\draw   (20,273) -- (590,273) -- (590,293) -- (20,293) -- cycle ;
\draw [line width=3]    (60,273) -- (60,293) ;
\draw [line width=3]    (160,273) -- (160,293) ;
\draw [line width=3]    (260,273) -- (260,293) ;
\fill[color={rgb, 255:red, 255; green, 133; blue, 0 }]   (162.2,273.5) -- (179.5,273.5) -- (179.5,292.5) -- (162.2,292.5) -- cycle ;
\draw    (180,273) -- (180,293) ;
\fill[color={rgb, 255:red, 39; green, 255; blue, 0 }]   (180.5,273.5) -- (200,273.5) -- (200,292) -- (180.5,292) -- cycle ;
\draw    (200,273) -- (200,293) ;
\fill[color={rgb, 255:red, 255; green, 246; blue, 100 }]   (200.5,273.5) -- (219,273.5) -- (219,292) -- (200.5,292) -- cycle ;
\draw    (220,273) -- (220,293) ;
\draw    (240,273) -- (240,293) ;
\fill[color={rgb, 255:red, 0; green, 228; blue, 255 }]   (240.5,273.5) -- (258,273.5) -- (258,292) -- (240.5,292) -- cycle ;
\draw [line width=3]    (350,273) -- (350,293) ;
\draw [line width=3]    (450,273) -- (450,293) ;
\draw    (520,273) -- (520,293) ;
\draw    (550,273) -- (550,293) ;
\draw    (560,273) -- (560,293) ;
\draw    (210,243) -- (210,263) ;
\draw [shift={(210,265)}, rotate = 270] [color={rgb, 255:red, 0; green, 0; blue, 0 }  ][line width=0.75]    (10.93,-3.29) .. controls (6.95,-1.4) and (3.31,-0.3) .. (0,0) .. controls (3.31,0.3) and (6.95,1.4) .. (10.93,3.29)   ;
\draw    (555.26,243) -- (555.26,269) ;
\draw [shift={(555,271)}, rotate = 270] [color={rgb, 255:red, 0; green, 0; blue, 0 }  ][line width=0.75]    (10.93,-3.29) .. controls (6.95,-1.4) and (3.31,-0.3) .. (0,0) .. controls (3.31,0.3) and (6.95,1.4) .. (10.93,3.29)   ;
\draw    (210,323) -- (210,297) ;
\draw [shift={(210,295)}, rotate = 90] [color={rgb, 255:red, 0; green, 0; blue, 0 }  ][line width=0.75]    (10.93,-3.29) .. controls (6.95,-1.4) and (3.31,-0.3) .. (0,0) .. controls (3.31,0.3) and (6.95,1.4) .. (10.93,3.29)   ;

\draw (21,351) node [anchor=north west][inner sep=0.75pt]   [align=left] {Target vector:};
\draw (37,382.4) node [anchor=north west][inner sep=0.75pt]  [font=\normalsize]  {$\vec{t}$};
\draw (248,384.4) node [anchor=north west][inner sep=0.75pt]  [font=\normalsize]  {$0$};
\draw (517,384.4) node [anchor=north west][inner sep=0.75pt]  [font=\normalsize]  {$1$};
\draw   (20,380) -- (590,380) -- (590,400) -- (20,400) -- cycle ;
\draw [line width=3]    (60,380) -- (60,400) ;
\draw [line width=3]    (450,380) -- (450,400) ;

\end{tikzpicture}
}
\caption{A simplified pictorial illustration of our main construction. For detailed illustration see Figure~\ref{figure for gap creating}.}
    \label{Pictorial illustration in introduction}
\end{figure}
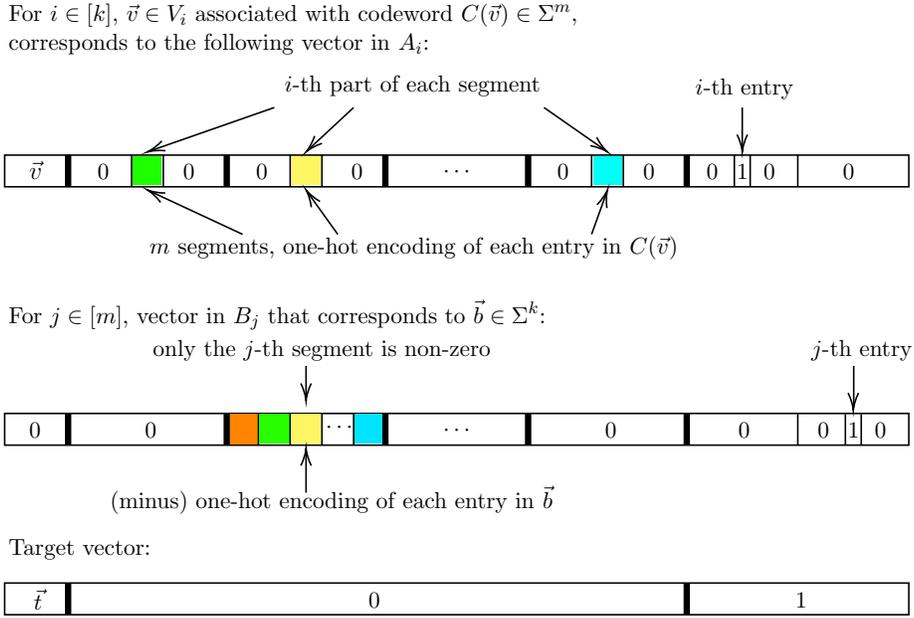

Below we illustrate the idea of our reduction. For simplicity, here we consider $k$-MLD problem over binary field. 
Given $k$ vectors sets $V_1,\ldots,V_k\subseteq\mathbb{F}_2^d$, a target vector $\vec t$ and a strong threshold graph $T=(A\dot\cup B,E_T)$, we first identify  $V_i$ with $A_i$ for every $i\in[k]$. Our goal is to construct a one-to-one mapping $f :A\cup B\to\mathbb{F}_2^D$ and a new target vector $\vec t'\in\mathbb{F}_2^D$ for some $D=\mathsf{poly}(d,k)$ such that in order to pick vectors from $f(A\cup B)$\footnote{Here we let $f(X)$  denote the set $\{f(x) : x\in X\}$.} with their sum being $\vec t'$,  one has to pick a set $f(X)$ of vectors  from $f(A)$ for some $X\subseteq A$ with $\sum_{\vec a\in X}\vec a=\vec t$ and a set $f(Y)$ of vectors  from $f(B)$ for some $Y\subseteq B$ such that for every $i\in[m]$,
\begin{description}
    \item[(a)] \label{Gap Property (a)}either $|Y\cap B_i|\ge 2$,
    \item[(b)] \label{Gap Property (b)}or $|Y\cap B_i|=1$ and there exists $b_i\in B_i$ with one of following properties:
\begin{description}
    \item[(b.1)]  $|X|=k$ and $b_i$ is the common neighbors of vertices in $X$.
    \item[(b.2)]  $b_i$ has at least $k+1$ neighbors in $X$.
\end{description}
\end{description}
Then we argue that these properties imply a constant gap between the solution sizes in the (YES) and (NO) cases of the $k$-MLD problem.
\begin{description}

    \item [(YES)] Suppose there are $a_1\in A_1,\ldots,a_k\in A_k$ such that $\sum_{i\in[k]}a_i=\vec t$. By the property (iii) of threshold graphs,  $a_1,\ldots,a_k$ have a common neighbor $b_i\in B_i$ for every $i\in[m]$. Then according to (b), the sum of $f(a_1),\ldots,f(a_k)$ and $f(b_1),\ldots,f(b_m)$ is $\vec t'$.

\item [(NO)] On the other hand, if  there are no $a_1\in A_1,\ldots,a_k\in A_k$ such that $\sum_{i\in[k]}a_i=\vec t$, then  one should pick either at least $(1-\varepsilon)2m$ vectors from $f(B)$ and $k+1$ vectors from $f(A)$, or pick a subset of vectors $f(X)$ from $f(A)$ and a subset of vectors $f(Y)$ from $f(B)$ for some $Y\subseteq B$ with $|\{i\in [m] : |Y\cap B_i|= 1\}|\ge \varepsilon m$. Let $I=\{i\in [m] : |Y\cap B_i|= 1\}$. 
According to the property (b.2), each vertex in $Y\cap B_i$ ($i\in I$) has $k+1$ neighbors in $X$. Since $|I|\ge\varepsilon m$, by the property (iv) of threshold graphs, we have that $|X|>h$. Thus, either  $(1-\varepsilon)2m$ vectors in $f(B)$ and $k+1$ vectors in $f(A)$  or  $m$ vectors in $f(B)$  and $h$ vectors in $f(A)$ must be picked in this case. 
\end{description}
To obtain a constant gap, we duplicate each vector in $f(A)$ $m/k$ times and let $h=ck$ where $c$ is some constant to be chosen. In the (yes) case, there are $2m$ vectors with their sum being $\vec t'$. In the (no) case, no $\min\{2(1-\varepsilon)m+m,m+cm\}$ vectors from $f(A\cup B)$ can have sum $\vec t'$.

The proof framework above has two problems to be solved. 
\begin{description}
    \item[(P1)] How to combine the threshold graph and the $k$-MLD instance to produce vectors $f(A\cup B)$ with the properties (a) and (b)?
    \item[(P2)] The smaller parameter blow-up we create in reduction, the tighter running time lower bound we obtain. So how to construct a threshold graph with $h>ck$ and $m$ as small as possible?
\end{description}

\noindent\textbf{Our approach to solve Problem (P1).} Problem (P1) is related to the composition step in the threshold graph composition method. For the \textsc{$k$-SetCover} problem, we can use the hypercube partition system~\cite{Feige98} to solve this problem. Unfortunately, this does not apply to the $k$-MLD problem. To solve problem (P1), we  exploit an additional property from the  construction of strong threshold graph using error correcting codes. More precisely, we can assume that there is a encoding function $C : A\to \Sigma^m$ and each $b_i\in B_i$ can be written as a $k$-tuple in $(b_{i,1},\ldots,b_{i,k})\in \Sigma^k$ such that $b_i$ is adjacent to $a_j\in A_j$ in the threshold graph if and only if $b_{i,j}=C(a_j)[i]$. Informally speaking, we choose the target vector $\vec t'$ and the one-to-one mapping $f:A\cup B\to\mathbb{F}_2^D$ such that any subset of vectors in $f(A\cup B)$ summing up to $\vec t'$ must contains, for each $i\in[m]$, at least one vector $f(b_i)$ for some $b_i\in B_i$. And if it contains exactly only one such vector $f(b_i)$, then one need to pick at least $k$ vectors $f(a_1)\in f(A_1),\ldots,f(a_k)\in f(A_k)$  to cancel out the parts corresponding to $b_{i,1},\ldots,b_{i,k}$ in the vector $f(b_i)$. A careful analysis shows that this construction has the properties (a) and (b).

\medskip

\noindent\textbf{Our approach to solve Problem (P2).} 
The construction of strong threshold graph in~\cite{LRSW23} was based on the idea of Karthik and Navon \cite{KN21}. Karthik and Navon \cite{KN21} observed that the ``collision number'' of an error-correcting code can be directly used to show the threshold property. Intuitively speaking, a set $C$ of strings with high $\varepsilon$-collision number indicates that if there is some mechanism forces us to choose some strings in $C$ that collides on at least $\varepsilon$ fraction of entries, then we must choose at least $\mathrm{Col}_{\varepsilon}(C)$ strings.

Known analysis of collision number in \cite{KN21,BCGR23} starts from the distance of an error-correcting code. For a code with relative distance $\delta$, previous analysis shows that its $\varepsilon$-collision number is $\mathrm{Col}_{\varepsilon}(C)=\sqrt{\frac{2\varepsilon}{1-\delta}}$.  Note that $\delta=1-\Theta(\frac{r}{m})$ for Reed-Solomon codes used in the previous works. To obtain a gap, we require $\mathrm{Col}_{\varepsilon}(C)\ge\Theta(k)$, which leads to $m=\Omega(k^2)r$. In our reduction, we additionally require $\Sigma^r\geq n$ to fit the input size, which requires $r\geq \frac{\log n}{\log |\Sigma|}$, then we have $m\geq k^2\log n/\log|\Sigma|$. On the other hand, our reduction needs to enumerate every $k$-tuples in $\Sigma^k$, concerning the running time we require $|\Sigma|^k\le n^{O(1)}$. Putting all together, we must have $m\ge \Omega(k^3)$. In fact, we showed that the Singleton bound of codes implies such construction must have parameter growth $\Omega(k^3)$.

To obtain a better parameter, we find the analysis by Karthik and Navon \cite[Section 3.1]{KN21} can be modified to show better lower bound for the $\varepsilon$-collision number of a random code. Following their idea, we show a random code $C_R:\Sigma^r\rightarrow \Sigma^m$ with superconstant-sized alphabet and $m=\Omega(|\Sigma|^{1/3}\log |\Sigma| r)$ would have  $\varepsilon$-collision number $\mathrm{Col}_{\varepsilon}(C_R)\geq |\Sigma|^{1/3}$, with high probability. Setting $|\Sigma|=\Theta(k^3)$, we have $\mathrm{Col}_{\varepsilon}(C)\ge\Theta(k)$. But now the parameter $m=\Omega(|\Sigma|^{1/3}\log |\Sigma| r)\ge k\log n$ is too large. Our solution is to consider a new error correcting code with small  dimension by increasing the alphabet size and show that this new code has the same collision number. More precisely, we partition the $m$ bits into $g$ blocks, each containing $m/g$ bits and treat the code words as strings in $\Sigma'^g$ where $\Sigma'=\Sigma^{m/g}$. 
Since $|\Sigma'^k|\le n^{O(1)}$, we have $m/g\le O(\frac{\log n}{k\log |\Sigma|})=O(\frac{\log n}{3k\log k})$. Thus, $g\ge \Theta(\frac{mk\log k}{\log n})\ge \Theta(k^2\log k)$. 
This reduces the  parameter growth from $k^3$ to $k^2\log k$, and the (randomized) ETH-based running time lower bound can be improved to $n^{O(\sqrt{k/\log k})}$. We hope to see whether some better construction of threshold graph leads to better lower bound of problems we discuss. 

\subsection{Previous Work}\label{sec:prework}
The parameterized complexity of \textsc{$k$-MDP} had been open for many years. This problem was first resolved by~\cite{BBE+21}. Interestingly, the reduction in \cite{BBE+21} also ruled out constant $\fpt$-approximation algorithm for \textsc{$k$-MDP} over binary field. In addition, they also ruled out any constant $\fpt$-approximation algorithm for \textsc{$k$-CVP} in all $\ell_p$ norms. Recent work by Bennett, Cheraghchi, Guruswami and Ribeiro \cite{BCGR23} proved parameterized inapproximability for 
\textsc{$k$-MDP} over all finite fields and \textsc{$k$-SVP} in all $\ell_p$ norms and arbitrary constant gap. 
These results are all based on the $\wone$-hardness of constant \textsc{Gap-$k$-NCP} or \textsc{Gap-$k$-CVP} in \cite{BBE+21}. 

Unfortunately, the gap-creating reduction from \textsc{$k$-Clique} to constant \textsc{Gap-$k'$-NCP} or \textsc{Gap-$k'$-CVP} in \cite{BBE+21} has a long reduction chain and causes a significant increase in the parameter. For example,  the reduction from \textsc{$k$-Clique} to constant \textsc{Gap-$k'$-NCP}  contains the following steps (the reduction for Gap-$k'$-CVP is similar):
\begin{itemize}
    \item The first step is to reduce \textsc{$k$-Clique} to the \textsc{One-Sided Gap Biclique} problem. In this step, the reduction outputs a bipartite graph $H=(L\cup R,E)$ and three integers $s=k(k-1)/2$, $\ell=(k+1)!$ and $h>\ell$ on input a graph $G$ and an integer $k$ such that if $G$ contains a $k$-clique, then there are $s$ vertices in $L$ with $h$ common neighbors. On the other hand, if $G$ contains no $k$-clique, then every $s$-vertex set of $L$ has at most $\ell$ common neighbors in $R$.
    \item The second step is to reduce the \textsc{One-Sided Gap Biclique} problem to \textsc{Gap-$k'$-Linear Dependent Set} problem (\textsc{Gap-$k'$-LDS})\footnote{In fact, the reduction in~\cite{BBE+21} from \textsc{One-Sided Gap Biclique} to \textsc{Gap-$k$-LDS} goes though an intermediate problem called gap bipartite subgraph with minimum degree (\textsc{GapBSMD}).}. On input the bipartite graph $H=(L\cup R,E)$ and three positive integers $s,\ell,h\in\mathbb{N}$, the reduction outputs a set $W$ of vectors and an integer $k'=hs$ such that, if $H$ contains a $K_{s,h}$-subgraph, then there are $k'$ vectors in $W$ that are linearly dependent. If every $s$-vertex set in $L$ has at most $\ell$ common neighbors, then any linearly dependent set of $W$ must have size at least $(h/\ell)^{1/s}$. To create a constant gap, one must choose a large parameter $h$  such that $(h/\ell)^{1/s}\ge \gamma hs$ for some $\gamma>1$. 
    Hence in ~\cite{BBE+21}, the authors have to set  $h=(k+6)! \cdot (\gamma k^2)^{k^2}$ and $k'=hs\ge   k^{\Omega(k^2)} = 2^{\Omega(k^2\log k)}$.
    \item The next step is to reduce the \textsc{Gap-$k'$-Linear Dependent Set} problem (\textsc{Gap-$k'$-LDS}) to \textsc{Gap-$k''$-Maximum Likelihood Decoding} problem (\textsc{Gap-$k''$-MLD})\footnote{Again, they introduced a color-coding technique to \textsc{Gap-$k$-LDS} (\textsc{Gap-$k$-Colored-LDS}) and used it as an intermediate problem between \textsc{Gap-$k$-LDS} and \textsc{Gap-$k$-MLD}, for details see \cite[Lemma 4.8, Theorem 5.4]{BBE+21}.}. This reduction preserves the parameter i.e., $k''=k$. 
    \item The remaining step gives a reduction from constant \textsc{Gap-$k''$-MLD} to constant \textsc{Gap-$k''$-NCP}.
\end{itemize}
Combining this with the $f(k)\cdot n^{\Omega(k)}$-time lower bound for the \textsc{$k$-Clique} problem, we only get a $g(k)\cdot n^{\Omega((\log k)^{1/(2+\epsilon)})}$-time lower bound for \textsc{Gap-$k$-NCP} using the reduction from~\cite{BBE+21}. 

Under a stronger gap assumption (Gap-ETH), Manurangsi~\cite{Man20} showed a tight $n^{\Omega(k)}$ time lower bound for constant approximating problems discussed in this article. His approach is to show an $n^{\Omega(k)}$ time lower bound for constant approximating \textsc{LaberCover}, then reduce it to  \textsc{$k$-UniqueSetCover}, then reduce \textsc{$k$-UniqueSetCover} to gap problems we discuss using reduction in~\cite{ABSS97}. The key step in his proof is to establish hardness result for approximating \textsc{$k$-UniqueSetCover}. 
To our best knowledge, there is no hardness of approximation result for the parameterized \textsc{$k$-UniqueSetCover} under gap-free assumptions, e.g. ETH and $\WOne\neq \FPT$.

Very recently, Guruswami, Ren and Sandeep~\cite{GRS23} showed constant \FPT-inapproximability of \textsc{$k$-UniqueSetCover} under the assumption that Average Baby PIH holds even for \textsc{2CSP} instance having rectangular relations. 
It's interesting whether their result and method can shed some light on showing ETH-based $n^{\Omega(k)}$ time lower bound for \textsc{$k$-UniqueSetCover}. 
We remark that the ETH-based $n^{\Omega(k)}$ time lower bound for constant approximating \textsc{$k$-UniqueSetCover} is still an open problem, and so does its \FPT-inapproximability  assuming $\WOne\neq \FPT$. 

\subsection{Paper Organization}
In Section~\ref{Section: Prelim}, we give preliminary of this paper. In Section~\ref{Section: Collision Number}, we give a new analysis on collision number of random code, this section can be skipped if readers wants to see the reduction directly. In Section~\ref{sec:gap creating reduction}, we present our gap-creating reduction for \textsc{$k$-MLD$_p$}. In Section~\ref{section: from NCP}, we show how to apply our reduction to other results and show inapproximability of other problems. For self-containment, we give a proof of equivalence between \textsc{$k$-MLD$_p$} and \textsc{$k$-NCP$_p$} in Appendix~\ref{appendix: equivalence between MLD and NCP}. 
\section{Preliminaries}
\label{Section: Prelim}
For integer $m>0$, let $[m]=\{1,2,\cdots, m\}$. For prime power $p>1$, we let $\mathbb{F}_p=\{0,1\cdots, p-1\}$ denote the $p$-sized finite field. We denote $\F_p^+$ as $\F_p\backslash\{0\}$. For a vector $\vec{v}\in \Sigma^m$ and $i\in[m]$, let $\vec{v}[i]\in\Sigma$ denote the $i$-th entry of $\vec{v}$. For two vectors $\vec{u}$, $\vec{v}$, let $\vec{u}\circ \vec{v}$ denote their concatenation. The symbol $\dot\cup$ denotes for the union set of multiple disjoint sets. As a supplement of big-$O$ notation, we let $f(k,n)=O_k(g(n))$ denote there exists constant $c>0$ and computable function $h:\mathbb{N}\rightarrow \mathbb{N}$ such that for any fixed $k>0$, $f(k,n)<c\cdot h(k)g(n)$ holds for all sufficiently large $n$.

For alphabet $\Sigma$ and vector $\vec{u},\vec{v}\in\Sigma^m$, the relative distance of them is defined as $\mathsf{dist}(\vec{u},\vec{v})=\frac{|\{i\in[m]:\vec{u}[i]\neq\vec{v}[i]\}|}{m}$. In this article, we sometimes use ``distance'' as shorthand of relative distance. For vector $\vec{v}\in\mathbb{Z}^m$ and $p\geq 1$, let the $\ell_p$ norm of $\vec{v}$ be $\ell_p(\vec{v})=(\Sigma_{1\leq i\leq m}|\vec{v}[i]|^p)^{1/p}$.
\subsection{Error-correcting Codes}
Error-correcting code plays a fundamental role in computer science and information theory. The problem we mainly discuss in this article and the construction we use are closely related to them. We give a general definition of error-correcting code. 
A detailed and systematic introduction to coding theory can be found at \cite{EssentialCodingTheory}.
\begin{definition}[Error-correcting Codes]
    Fix an alphabet $\Sigma$, an error-correcting code with length $m$ and relative distance $\delta>0$ is a subset $\mathcal{C}\subseteq\Sigma^m$ satisfying for all $\vec{x},\vec{y}\in \mathcal{C}$, if $\vec{x}\neq \vec{y}$, $\mathsf{dist}(\vec{x},\vec{y})\geq \delta$.
\end{definition}

In the study of coding theory, when considering coding problems that related to decoding or distance, we usually restrict it to \textit{linear} codes. We give the definition of linear codes as follows.
\begin{definition}[Linear Codes]
    Fix an alphabet $\Sigma$ such that $\Sigma^r$ and $\Sigma^m$ being linear spaces, a linear code is an error-correcting code $\mathcal{C}\subseteq \Sigma^m$ associated with a linear function $f:\Sigma^r \rightarrow \Sigma^m$ that for all $x\in \Sigma^r$, $f(x)\in \mathcal{C}$.
\end{definition}
\subsection{Hypothesis}
We introduce the Exponential Time Hypothesis in this section. First, let us introduce the fundamental problem in computational complexity: the $3$-satisfiability problem \textsc{$3$-SAT}.
\begin{definition}[\textsc{3-SAT}]
    Given a \textsc{$3$-CNF} formula (conjunctive
formal form, each clause contains exactly $3$ literals) $\varphi$ with $n$ variables and $m$ clauses, decide if there exists a boolean assignment $z\in\{0,1\}^n$ that satisfies $\varphi$, i.e., $\varphi(z)=1$.
\end{definition}

The Exponential Time Hypothesis (ETH) states that \textsc{$3$-SAT} cannot be solved in subexponential time, formally:
\begin{hypothesis}[Exponential Time Hypothesis\cite{IP01,IPZ01}]
    There exists constant $\delta>0$ such that \textsc{$3$-SAT} with $n$ variable and $O(n)$ clauses cannot be solved in time $O(2^{\delta n})$.
\end{hypothesis}

Similarly, for randomized algorithms, the Randomized Exponential Time Hypothesis states that \textsc{$3$-SAT} cannot be solved by randomized algorithm in subexponential time, formally:
\begin{hypothesis}[Randomized Exponential Time Hypothesis]
    There exists constant $\delta>0$ such that \textsc{$3$-SAT} with $n$ variable and $O(n)$ clauses cannot be solved by randomized algorithm in time $O(2^{\delta n})$.
\end{hypothesis}
\subsection{Problems}
\label{prelim:probs}
We first give the definition of general parameterized Maximum Likelihood Decoding problem. 
\npprobyn[9.5]{$\gamma$-Gap-$k$-MLD$_p$}
{A vector multi-set $V\subseteq \F_p^d$ with size $n$ and a target vector $\vec t\in\F_p^d$}
{$k$}
{There exists $k$ distinct vectors (with respect to multi-set), $\vec v_1,\cdots,\vec v_k\in V$ and $\alpha_1,\dots,\alpha_l\in\F_p^+$ such that $\alpha_1\vec v_1+\dots+\alpha_k\vec v_k=\vec t$}
{Any $\ell\leq \gamma k$, $l$ vectors $\vec v_1,\dots,\vec v_l\in V$ and $\alpha_1,\dots,\alpha_l\in\F_p^+$ satisfies $\alpha_1\vec v_1+\dots+\alpha_l\vec v_l\neq\vec t$}

To fit our reduction, we start from a special restricted type of parameterized Maximum Likelihood Decoding problem that vectors are partitioned into $k$ different sets, and the YES case asks for selecting one vector from each set such that they directly add up to the target vector. This type of parameterized Maximum Likelihood Decoding problem  is formally defined as:
\npprobyn[9.5]{$\gamma$-Gap-$k$-ColoredMLD$_p$}
{$k$ vector multi-sets $V_1,\dots,V_k\subseteq \F_p^d$ each of size $n$ and a target vector $\vec t\in\F_p^d$}
{$k$}
{There exists $\vec v_1\in V_1,\dots,\vec v_k\in V_k$ such that $\vec v_1+\dots+\vec v_k=\vec t$}
{For any $\ell\leq \gamma k$, $\vec v_1,\dots,\vec v_l\in V_1\cup\dots\cup V_k$ and $\alpha_1,\dots,\alpha_l\in\F_p^+$ must satisfy $\alpha_1\vec v_1+\dots+\alpha_l\vec v_l\neq\vec t$}
The equivalence between \textsc{$\gamma$-Gap-$k$-MLD$_p$} and \textsc{$\gamma$-Gap-$k$-ColoredMLD$_p$} can be shown by creating $p$ new vectors corresponds to $p$ different coefficients for each of the original vector, then making $k$ copies of the vector set in one direction, and a standard color-coding technique in the other direction. Due to the equivalence, we shall omit the \textsc{Colored} script in the article and confuse these definitions to simplify the notations. In particular, we use \textsc{$k$-MLD$_p$} to denote \textsc{$\gamma$-Gap-$k$-ColoredMLD$_p$} when $\gamma=1$.

We next give the definition of parameterized \textsc{NCP} problem.
\npprobyn[9.5]{$\gamma$-Gap-$k$-NCP$_p$}
{A vector set $V=\{\vec{v}_1,\cdots,\vec{v}_n\}\subseteq \F_p^d$ with size $n$, a target vector $\vec t\in\F_p^d$}
{$k$}
{There exists $c_1,\cdots, c_n\in\mathbb{F}_p$ and $\vec{w}\in\mathbb{F}_p^d$ with $||w||_0\leq k$ such that $c_1\vec v_1+\cdots+c_n\vec v_n+\vec{w} = \vec{t}$}
{For any $c_1,\cdots, c_n\in\mathbb{F}_p$, $\vec{w} = \vec{t}-(c_1\vec v_1+\cdots+c_n\vec v_n)$ satisfies $||\vec{w}||_0>\gamma k$}

The homogeneous version of \textsc{NCP} is known as the parameterized \textsc{Minimum Distance Problem} as follows.
\npprobyn[9.5]{$\gamma$-Gap-$k$-MDP$_p$}
{A vector set $V=\{\vec{v}_1,\cdots,\vec{v}_n\}\subseteq \F_p^d$ with size $n$}
{$k$}
{There exists not all zero $c_1,\cdots, c_n\in\mathbb{F}_p$ satisfying $||c_1\vec v_1+\cdots+c_n\vec v_n||_0\leq k$}
{For all not all zero $c_1,\cdots, c_n\in\mathbb{F}_p$, $||c_1\vec v_1+\cdots+c_n\vec v_n||_0> \gamma k$}

There are two fundamental lattice problems that is closely related to coding problems above, we introduce them as follows. 
The first problem is the parameterized \textsc{Closest Vector Problem}, which asks if a given lattice is ``close to'' a target vector.
\npprobyn[9.5]{$\gamma$-Gap-$k$-CVP$_p$}
{A vector set $V=\{\vec{v}_1,\cdots,\vec{v}_n\}\subseteq \mathbb{Z}^d$ with size $n$, a target vector $\vec t\in\mathbb{Z}^d$}
{$k$}
{There exists $c_1,\cdots, c_n\in\mathbb{Z}$ and $\vec{w}\in\mathbb{Z}^d$ with $||w||_p^p\leq k$ such that $c_1\vec v_1+\cdots+c_n\vec v_n+\vec{w} = \vec{t}$}
{For any $c_1,\cdots, c_n\in\mathbb{Z}$, $\vec{w} = \vec{t}-(c_1\vec v_1+\cdots+c_n\vec v_n)$ satisfies $||\vec{w}||_p^p>\gamma k$}

The homogeneous version of \textsc{CVP} is known as the parameterized \textsc{Shortest Vector Problem} as follows.
\npprobyn[9.5]{$\gamma$-Gap-$k$-SVP$_p$}
{A vector set $V=\{\vec{v}_1,\cdots,\vec{v}_n\}\subseteq \mathbb{Z}^d$ with size $n$}
{$k$}
{There exists not all zero $c_1,\cdots, c_n\in\mathbb{Z}$ satisfying $||c_1\vec v_1+\cdots+c_n\vec v_n||_p^p\leq k$}
{For all not all zero $c_1,\cdots, c_n\in\mathbb{Z}$, $||c_1\vec v_1+\cdots+c_n\vec v_n||_p^p> \gamma k$}

\subsection{Probability Inequality}
We use the following Chernoff bound of random variables.
\begin{theorem}[Chernoff Bound]
    Consider independent random variables $X_1,\dots,X_n\in\{0,1\}$ with $X=\sum_{i=1}^{m}X_i$ and $\mu=\mathbb{E}[X]$. For any $0<\delta<1$ we have
    $$
        \Pr[X\leq (1-\delta)\mu]\leq \exp\left(-\frac{\mu\delta^2}{2}\right).
    $$
\end{theorem}

\section{Collision Number of Error Correcting Codes}
\label{Section: Collision Number}

Consider a collection of strings $S\subseteq \Sigma^m$, we say that $S$ ``collides'' on the $i$-th coordinate if there are distinct $x,y\in S$ such that $x[i]=y[i]$.
Following the work of \cite{KN21,LRSW23}, we define the collision number of a set of strings as follows.

\begin{definition}[$\varepsilon$-Collision Number]\label{definition: collision number}
    For a set $C\subseteq \Sigma^m$ and $0<\varepsilon< 1$, the $\varepsilon$-collision number of $C$, denote as $\mathrm{Col}_{\varepsilon}(C)$, is the smallest integer $s\in\mathbb{N}^+$ such that there exists $S\subseteq C$ with $|S|=s$ and $S$ collides on more than $\varepsilon m$ coordinates, i.e.,
    $$
        |\{i\in [m]\mid \exists x,y\in S, x\neq y \text{ s.t. } x[i]=y[i]\}|>\varepsilon m.
    $$
\end{definition}

Note from the definition that for any $S\subseteq C$, if $S$ collides on more than $\varepsilon m$ coordinates, then $|S|\geq \mathrm{Col}_{\varepsilon}(C)$.
An error-correcting code over alphabet $\Sigma$ can be viewed as a special subset of $\Sigma^m$ where $m$ is the length of codeword, so the definition above naturally applies to error-correcting codes.
A deterministic construction of error-correcting codes with high $\varepsilon$-collision number can be find in \cite{KN21,LRSW23}. Their construction does not directly obtain high collision number of a code, instead they showed implication from code distance to its ($\varepsilon$-)collision number as follows. 
\begin{lemma}[\cite{KN21}, See also Theorem 10 in \cite{LRSW23}]\label{lemma: collision number from ECC}
    For any constant $0<\varepsilon\leq 1$, an error correcting code $C:\Sigma^r\to\Sigma^m$ with relative distance $0<\delta<1$ has $\mathrm{Col}_{\varepsilon}(C)\geq\sqrt{\frac{2\varepsilon}{1-\delta}}$.
\end{lemma}

For Reed-Solomon codes, considering their distance, the following result is an immediate consequence.

\begin{theorem}[\cite{KN21,LRSW23}]
\label{thm:collision number in rs code}
    Fix any Reed-Solomon code $\mathcal{C}^{RS}:\Sigma^r\rightarrow \Sigma^m$ with $r<m\leq|\Sigma|$. For any $k\in\mathbb{N}$ and $0<\varepsilon<1$, $\mathrm{Col}_{\varepsilon}(\mathcal{C}^{RS})\geq \sqrt{\frac{2\varepsilon m}{r}}$.
\end{theorem}

To fit the requirement in our reduction, i.e., $|\Sigma|^r\geq n$, we choose $|\Sigma|=n^{1/k}$ and $r=\Omega(k)$. To fit the requirement that $\mathrm{Col}_{\varepsilon}(C)=\Omega(k)$ in Lemma~\ref{lemma: gap creating (warm up)}, the Reed-Solomon code here must satisfy $m=\Omega(k^2r)=\Omega(k^3)$. Seeking for a shorter code with high $\varepsilon$-collision number, we turn to randomized construction of codes, and we show the following lemma.

\begin{lemma}\label{lemma: collision number from random code}
    For any constant $0<\varepsilon<1$ and any random code $C_R:\Sigma^r\rightarrow \Sigma^m$ where each codeword is selected uniformly at random in $\Sigma^m$, if $m\geq 16\frac{1}{\varepsilon^2}|\Sigma|^{1/3}\ln|\Sigma|r$ and $|\Sigma|=\omega(1)$, then with high probability, $\mathrm{Col}_{\varepsilon}(C_R)> |\Sigma|^{1/3}$.
\end{lemma}
\begin{proof}
    We show that the probability that $\mathrm{Col}_{\varepsilon}(C_R)\leq|\Sigma|^{1/3}$ is small. Note that the event ``$\mathrm{Col}_{\varepsilon}(C_R)\leq|\Sigma|^{1/3}$'' is equivalent to ``there exists $S\subseteq C_R$ with $|S|=|\Sigma|^{1/3}$ such that $S$ collides on more than $\varepsilon m$ coordinates''. Our target is to upper bound the probability of this event.
    
    First, fix any $S\subseteq C_R$ with $|S|=|\Sigma|^{1/3}$ and $i\in[m]$, we show that with high probability $S$ does not collide on the $i$-th coordinate. To be convenient, we list the elements in $S$ as $S=\{x_1,\cdots, x_{|S|}\}$. Recall that ``$S$ does not collide on the $i$th coordinate'' means that ``$x_1[i],\cdots, x_{|S|}[i]$ are all distinct''. For $1\leq j\leq |S|$, we define event $E_j$ as ``$x_j[i]$ is not in $\{x_c[i]\}_{1\leq c<j}$'', and the event above is also equivalent to $E_1 \wedge \cdots \wedge E_{|S|}$. We now lower bound its probability as:
    \begin{align*}
    &\Pr[x_1[i],\cdots, x_{|S|}[i] \text{ are all distinct}]\\
    =&\Pr[E_1 \wedge \cdots \wedge E_{|S|}]\\
    =&\Pr[E_1]\times\Pr[E_2|E_1]\times \cdots \times \Pr[E_{|S|}|E_1\wedge\cdots\wedge E_{|S|-1}]\\
    =&1\cdot \frac{|\Sigma|-1}{|\Sigma|}\cdots \frac{|\Sigma|-(|S|-1)}{|\Sigma|}\quad\quad (\text{each conditioned event reduces one feasible choice})\\
    \geq &\left(\frac{|\Sigma|-|\Sigma|^{1/3}}{|\Sigma|}\right)^{|\Sigma|^{1/3}}\\
    =&\left(1-\frac{1}{|\Sigma|^{2/3}}\right)^{|\Sigma|^{1/3}}\\
    =&\left(1-\frac{1}{|\Sigma|^{2/3}}\right)^{|\Sigma|^{2/3}\cdot \frac{1}{|\Sigma|^{1/3}}}\\
    \geq& \left(\frac{1}{4}\right)^{\frac{1}{|\Sigma|^{1/3}}}\\
    =& 1-o(1)
    \end{align*}
    where the last inequality holds from the fact that $(1-1/n)^n\geq 1/4$ when $n\geq 2$. Denote the above probability as $\Delta$.

    Secondly, fix any $S\subseteq C_R$ with $|S|=|\Sigma|^{1/3}$, we upper bound the probability of ``$S$ collides on more than $\varepsilon m$ coordinates''. Let $B_i$ be the indicating variable of ``$S$ does \textbf{not} collide on $i$th position'' and let $B$ denotes the number of positions that $S$ does \textbf{not} collide on, i.e., $B=\sum_{i=1}^{m}B_i$. The event ``$S$ collides on more than $\varepsilon m$ coordinates'' is equivalent to ``$B\leq (1-\varepsilon)m$''. The expectation of $B$ is
    \[
    \mathbb{E}[B]=\sum_{i=1}^{m}\mathbb{E}[B_i]= \Delta m.
    \]
 
    From the construction of random code, we can see that $B_1,\dots,B_m$ are independent.
    Applying Chernoff bound, we have:
    \begin{align*}
        \Pr[B\leq (1-\varepsilon)m]=&\Pr[B\leq \Delta m - (\Delta -1+\varepsilon)m]\\
        \leq & \exp\left(-\frac{(\Delta-1+\varepsilon)^2\Delta m}{2 \Delta^2}\right)\\
        =&\exp\left(-\frac{(\Delta-1+\varepsilon)^2}{2\Delta}m\right)\\
        \leq & \exp\left(-\frac{(\Delta-1+\varepsilon)^2}{2}m\right) && (\text{since }\Delta\leq 1)\\
        \leq & \exp\left(-\frac{1}{8}\varepsilon^2m\right) && (\text{since }\Delta-1+\varepsilon=\varepsilon-o(1)\geq \frac{1}{2}\varepsilon).
    \end{align*}
    
    There are at most $(|\Sigma|^r)^{|\Sigma|^{1/3}}$ subsets of $C_R$ with size $|\Sigma|^{1/3}$, so we take the union bound over all possible $S$'s as:
    \begin{align*}
        \Pr[\mathrm{Col}_\varepsilon(C_R)\leq |\Sigma|^{1/3}]&\leq (|\Sigma|^r)^{|\Sigma|^{1/3}}\cdot \exp\left(-\frac{1}{8}\varepsilon^2m\right)\\
        &=e^{|\Sigma|^{1/3}\ln|\Sigma|r-\frac{1}{8}\varepsilon^2m}\\
        &\leq e^{-|\Sigma|^{1/3}\ln|\Sigma|r}\\
        &=o(1)
    \end{align*}
    where the last inequality is due to $m\geq 16\frac{1}{\varepsilon^2}|\Sigma|^{1/3}\ln|\Sigma|r$. Therefore with high probability, $\mathrm{Col}_\varepsilon(C_R)> |\Sigma|^{1/3}$.
\end{proof}

\begin{lemma}\label{lemma: collision number from random code (large m)}
    For any constant $c>0$ and $0<\varepsilon<1$, there is a randomized algorithm that given integers $n,k\in \N^+$, constructs a code $C\subseteq \Sigma^m$ with parameters $|C|=n, |\Sigma|=O(k^3)$ and $m=O(k\log n)$ such that with high probability, $\mathrm{Col}_\varepsilon(C)>ck$. Moreover, the running time of this algorithm is $O(nm|\Sigma|)$.
\end{lemma}
\begin{proof}
    The running time analysis is obvious since a random code simply selects $n$ codewords, each codeword is simply selecting $m$ symbols from $\Sigma$ at random. Let $|\Sigma|=(ck)^3=O(k^3)$ and $r=\log n/\log |\Sigma|$ such that $|\Sigma|^r=n$. Let $m=16\frac{1}{\varepsilon^2}|\Sigma|^{1/3}\ln|\Sigma|r=O(k\log n)$. We construct a random code $C:\Sigma^r\to\Sigma^m$ where each codeword is chosen independently and uniformly at random from $\Sigma^m$. By Lemma \ref{lemma: collision number from random code}, $\mathrm{Col}_\varepsilon(C)>|\Sigma|^{1/3}=ck$ with high probability.
\end{proof}

\begin{remark}
    We remark that using an almost identical argument, Lemma~\ref{lemma: collision number from random code} can be extended to the case that for each integer $t\geq 3$, if $m>\Omega(|\Sigma|^{1/t}\log |\Sigma|r)$ and $|\Sigma|=\omega(1)$, then w.h.p., $\mathrm{Col}_{\varepsilon}(C_R)>|\Sigma|^{1/t}$. For constant $t>3$, setting $|\Sigma|=\Omega(k^t)$, Lemma~\ref{lemma: collision number from random code (large m)} can be extended to the case with same parameter but larger code alphabet.
\end{remark}

The following is a ``merge'' step in out reduction that enables us to enumerate the composition of a number of blocks over small alphabet, which turned out to be useful in reducing parameter growth.
\begin{lemma}\label{lemma: collision number from random code (grouped)}
    For any constant $c>1$ and $0<\varepsilon<1$, there is a randomized algorithm that given integers $n,k\in \N^+$, constructs a code $C\subseteq \Sigma^m$ with parameters $|C|=n, |\Sigma|=O(n^{1/k})$ and $m=O(k^2\log k)$ such that with high probability $\mathrm{Col}_\varepsilon(C)>ck$. Moreover, the running time of this algorithm is $O(k^2\log kn^{1+1/k})$.
\end{lemma}
\begin{proof}
    On input $n,k$, we first construct a code $C'\subseteq {(\Sigma')}^{m'}$ by Lemma~\ref{lemma: collision number from random code (large m)}, where $|C'|=n,|\Sigma'|=O(k^3)$, $m'=O(k\log n)$ and with high probability $\mathrm{Col}_\varepsilon(C')>ck$. Let $g$ be some integer to be determined later. The idea is to merge every $g$ coordinates of a codeword in into a single coordinate of the resulting codeword. To illustrate, we construct $C\subseteq \Sigma^m$ as follows.
    Let $\Sigma=(\Sigma')^g$ and $m=m'/g$. For every $c'\in C'$ we introduce a codeword $c$ into $C$ that for every $i\in[m]$
    $$
        c[i] = (c'[ig], c'[ig+1], \dots, c'[ig+g-1]).
    $$
    
    Now suppose that $\mathrm{Col}_\varepsilon(C')>h$, we prove in the following that $\mathrm{Col}_\varepsilon(C)>h$.
    Note that by the definition of collision number, it suffices to prove that: for any $S\subseteq C$, if $S$ collide on more than $\varepsilon m$ coordinates, then $|S|>h$.
    Now Suppose $S\subseteq C$ is such a set of codewords that collide on more than $\varepsilon m$ coordinates.
    Let $i\in[m]$ be one of these coordinates. Then there are distinct codewords $c_1, c_2\in S$ that $c_1[i]=c_2[i]$.
    Let $S'$ be the corresponding set of codewords in $C'$, and $c_1',c_2'\in S'$ be the corresponding codewords of $c_1,c_2$.
    By the way we construct, $c_1[i]=c_2[i]$ means that for all $ig\leq j<(i+1)g$, $c_1'[j]=c_2'[j]$.
    This means that $S'$ collides on all $g$ coordinates of $[ig,(i+1)g)$.
    Therefore, since $S$ collides on more than $\varepsilon m$ coordinates, $S'$ collides on more than $g\cdot \varepsilon m=\varepsilon m'$ coordinates.
    Since $\mathrm{Col}_\varepsilon(C')>h$, by the definition of collision number it must satisfy that $|S'|\geq \mathrm{Col}_\varepsilon(C')>h$.
    This means that $|S|=|S'|>h$ as well.

    Since we have proved that the collision number preserves through the ``merging process'', it holds with high probability that $\mathrm{Col}_\varepsilon(C)>ck$. Finally let $g=\frac{\log n}{k\log |\Sigma|'}$ then we achieve the desired parameters as $|\Sigma|=(|\Sigma'|)^g=O(n^{1/k})$ and $m=m'/g=O(k^2\log k)$.

    The running time follows from Lemma~\ref{lemma: collision number from random code (large m)}.
\end{proof}

\subsection{Limitation of Collision Analysis in \cite{KN21,LRSW23}}
There are two approaches to prove that a random code has good collision number. One is to prove directly as our approach in Lemma~\ref{lemma: collision number from random code}. The other is to first prove that a random code has good relative distance, then use the lower bound for collision number in Lemma \ref{lemma: collision number from ECC}. We have already shown the first approach yields $m'=O(k^2\log k)$. Below, we argue that the second approach must cause a cubic increase in the parameter.

To fit Lemma~\ref{lemma: gap creating (warm up)} in the following paragraph, we require the collision number of code $C$ be $\mathrm{Col}_{\varepsilon}(C)=\Omega(k)$. Combining with Lemma \ref{lemma: collision number from ECC}, we immediately have the relative distance of code $C$ must satisfies\[
\delta\geq 1-\frac{1}{\Omega(k^2)}.
\]

In coding theory, some bounds are established for parameters of a code. We introduce the Singleton bound of a code as follows.
\begin{theorem}[Singleton Bound]
    For every code $C:\Sigma^{r}\rightarrow\Sigma^m$ with relative distance $\delta$, $r\leq m-\delta m+1$.
\end{theorem}
Detailed discussion and proof of Singleton bound can be found in \cite[Section 4.3]{EssentialCodingTheory}. We apply the bound to parameter we choose and obtain $m-(1-\frac{1}{\Omega(k^2)}))m+1\geq r$, i.e., \[
m\geq \Omega(k^2)r.\]

Our reduction for \textsc{MLD} associates each input vector with a unique codeword, which requires $|C|\geq n$, or $|\Sigma|^r\geq n$, leading to \[
r\geq \frac{\log n}{\log |\Sigma|}.
\]

Finally, consider the ``merging'' procedure in Lemma~\ref{lemma: collision number from random code (grouped)}, we merge the code into $m'$ blocks, each new block contains $g=\frac{m}{m'}$ blocks, then the set of all $k$-tuples of a new block has size
\begin{align*}
    (|\Sigma|^{g})^k=& |\Sigma|^{k\frac{m}{m'}}\\
    =&2^{k\frac{m}{m'}\log |\Sigma|}\\
    \geq & 2^{\Omega(k\cdot k^2\frac{\log n}{\log|\Sigma|}\log|\Sigma|)/m'}\\
    =&n^{\Omega(k^3)/m'}.
\end{align*}
To efficiently enumerate all $k$-tuples of a new block, the size above must be at most polynomial in $n$, indicating that the final blocks $m'=\Omega(k^3)$. This bound is tight since we've shown the Reed-Solomon code can achieve $m'=O(m)=O(k^3)$. 

\section{Gap-creating Reduction for \textsc{$k$-MLD$_p$}}
\label{sec:gap creating reduction}

In this section we present our gap-creation reduction for \textsc{$k$-MLD$_p$}. First we present a construction that illustrates our main idea and is the crux of our reduction. 
This construction produces an ``unbalanced gap'' \textsc{$k'$-MLD$_p$} instance in the sense that the output instance is divided into two parts (with different sizes), any solution must contain an amount of vectors in each part. Further, for the NO case, any solution must contain constant fraction more vectors in at least one part.
This construction still needs to be modified later to convert into an actual reduction. 


\begin{lemma}\label{lemma: gap creating (warm up)}
There is an algorithm which on input $k$ vector sets $V_1,\cdots,V_k\subseteq\mathbb{F}_p^d$ each of size $n$, a target vector $\vec{t}\in\mathbb{F}_p^d$ and a code $C\subseteq|\Sigma|^m$ with $|C|=n$ and $\mathrm{Col}_\varepsilon(C)\geq ck$ outputs $A=A_1\dot\cup\cdots\dot\cup A_{k}\subseteq\mathbb{F}_p^D$ and $B=B_1\dot\cup\ldots\dot\cup B_m\subseteq\mathbb{F}_p^{D}$ with  $D=O(d+km|\Sigma|)$ and a target vector $\vec{t}'\in\F_p^D$ in 
$O(dm^2k^2|\Sigma|(n+|\Sigma|^k))$-time such that
\begin{itemize}
    \item[(i)]  If there exist $\vec{v}_1\in V_1,\ldots,\vec{v}_k\in V_k$ such that $\sum_{i\in[k]}\vec{v}_i=\vec{t}$, then there exists $\vec{a}_1'\in A_1,\cdots, \vec{a}_k'\in A_k$ and $\vec{b}_1'\in B_1,\cdots, \vec{b}_m'\in B_m$ with their sum being $\vec{t}'$.
    \item[(ii)] If for any $\vec v_1\in V_1,\dots,\vec v_k\in V_k$ and $\alpha_1,\dots\alpha_k\in\F_p^+$ it holds that $\alpha_1\vec v_1+\dots+\alpha_k\vec v_k\neq\vec t$, then any $X\subseteq A\dot\cup B$ and $\lambda:X\to\F_p^+$ such that $\sum_{\vec{x}\in X}\lambda(\vec x)\vec{x}=\vec{t}'$ must satisfy at least one of the following:
    \begin{itemize}
        \item $|X\cap A|\geq ck$ and $|X\cap B|\geq m$,
        \item $|X\cap A|\geq k$ and $|X\cap B|\geq 2(1-\varepsilon)m$.
    \end{itemize}
\end{itemize}
\end{lemma}

\begin{proof}
The resulting dimension is $D=d+mk|\Sigma|+k+m$. We break the resulting dimension into $4$ blocks respectively of size $d,mk|\Sigma|,k$ and $m$.
To be precise, for any vector $\vec x\in\F_p^D$, let
\begin{itemize}
    \item $\vec x^{(1)}\in\F_p^d$ be the first block,
    \item $\vec x^{(2)}\in\F_p^{mk|\Sigma|}$ be second block,
    \item $\vec x^{(3)}\in\F_p^k$ be the third block,
    \item $\vec x^{(4)}\in\F_p^m$ be the fourth block.
\end{itemize}
We further break the second block into $m$ sub-blocks each of size $k|\Sigma|$, i.e., $\vec x^{(2)}=\vec x^{(2,1)}\circ\dots\circ\vec x^{(2,m)}$.

We let $\vec e_i$ be the indicator vector of which the $i$-th entry is $1$ and the other entries are $0$.
To be convenient, the dimension of $\vec e_i$ depends on the context.
Specially we let $\iota :\Sigma\to [|\Sigma|]$ be an arbitrary bijection, and for every $\sigma\in\Sigma$ we let
\[
    \vec{e}_\sigma=(\underbrace{\overbrace{0,\dots,0}^{\iota(\sigma)-1},1,0\dots,0}_{|\Sigma|}).
\]

\noindent\textbf{Construction of $A$.}
For every $V_i$, associate each $\vec{v}\in V_i$ a distinct codeword of $C$, denoted by $C(\vec{v})$.
For every $i\in[k]$ and $\vec{v}\in V_i$, introduce a vector $\vec{a}_{i,\vec{v}}$ as
\begin{itemize}
    \item $\vec{a}_{i,\vec{v}}^{(1)}=\vec v$,
    \item $\vec{a}_{i,\vec{v}}^{(2,j)}=(\underbrace{\overbrace{\vec 0,\dots,\vec 0}^{(i-1)},\vec{e}_{C(\vec v)[j]},\vec 0,\dots,\vec 0}_{k})$, for every $j\in[m]$,
    \item $\vec{a}_{i,\vec{v}}^{(3)}=\vec e_i$,
    \item $\vec{a}_{i,\vec{v}}^{(4)}=\vec{0}_{m}$.
\end{itemize}
And we let 
\[
    A_i=\{\vec{a}_{i,\vec{v}}\mid \vec{v}\in V_i\} \text{ and } A = A_1\cup\dots\cup A_k.
\]

\noindent\textbf{Construction of $B$.}
For every $j\in[m]$ and $\vec{\sigma}=(\sigma_1,\dots,\sigma_k)\in\Sigma^k$, introduce a vector $\vec{b}_{j,\vec{\sigma}}$ as
\begin{itemize}
    \item $\vec{b}_{j,\vec{\sigma}}^{(1)}=\vec 0_d$,
    \item $\vec{b}_{j,\vec{\sigma}}^{(2,j)}=(-\vec e_{\sigma_1},\dots-\vec e_{\sigma_k})$,
    \item $\vec{b}_{j,\vec{\sigma}}^{(2,j')}=\vec 0_k$ for every $j'\in[m]\backslash\{j\}$,
    \item $\vec{b}_{j,\vec{\sigma}}^{(3)}=\vec 0_k$,
    \item $\vec{b}_{j,\vec{\sigma}}^{(4)}=\vec e_j$.
\end{itemize}
We let 
\[
    B_j=\{\vec{b}_{j,\vec{\sigma}}\mid \vec{\sigma}\in\Sigma^k\} \text{ and } B = B_1\cup\dots\cup B_m.
\]
Finally we set the target vector $\vec t'$ as
\begin{itemize}
    \item $\vec t'^{(1)}=\vec t$,
    \item $\vec t'^{(2)}=\vec 0_{mk|\Sigma|}$,
    \item $\vec t'^{(3)}=\vec 1_k$,
    \item $\vec t'^{(4)}=\vec 1_m$.
\end{itemize}

\begin{figure}[t]
    \centering
    \includegraphics[scale=0.4]{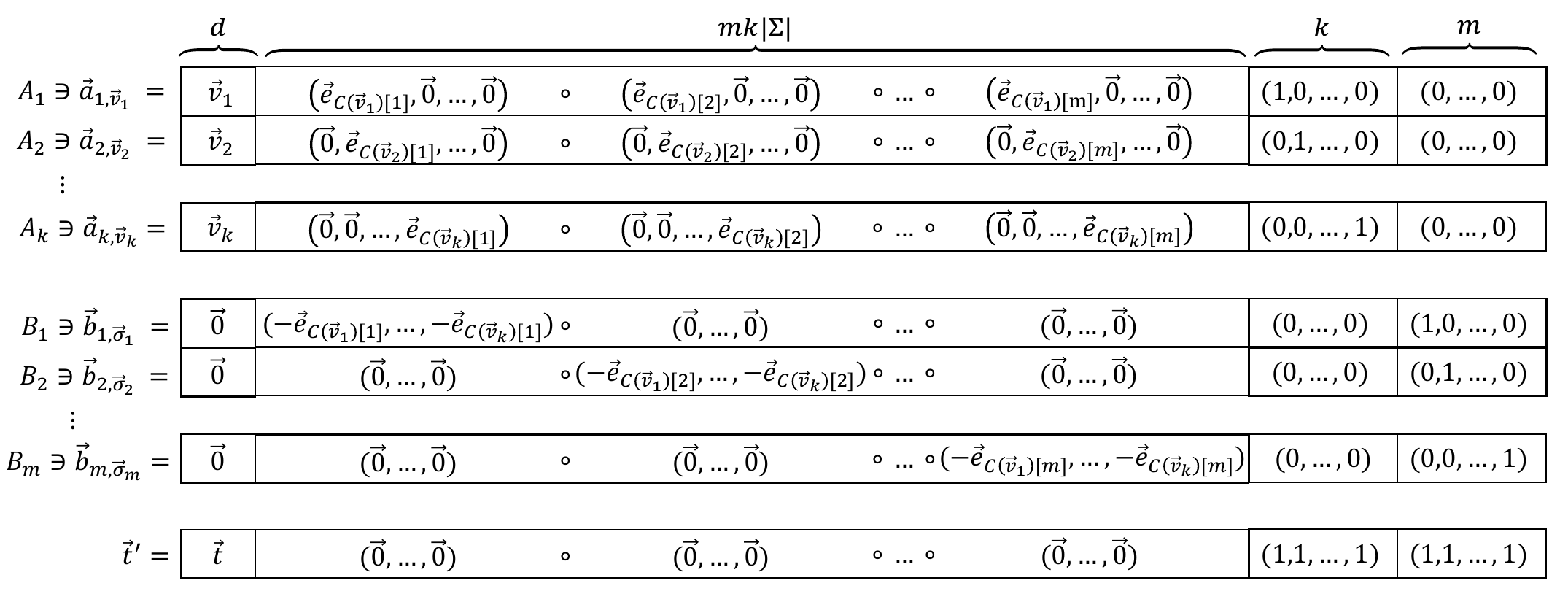}
    \caption{Illustration for the vectors of Lemma \ref{lemma: gap creating (warm up)} in the completeness setting. We can choose each $\vec{b}_{j,\vec{\sigma}_j}$ as $\vec{\sigma}_j=(C(\vec{v}_1)[j], \cdots, C(\vec{v}_k)[j])$.}
    \label{figure for gap creating}
\end{figure}

\noindent\textbf{Time complexity.}
Producing each vector in $A$ requires $O(d+mk|\Sigma|+km)=O(d+mk|\Sigma|)$ time, so the total time cost producing $A$ is $O(dkn+mk^2n|\Sigma|)$. Producing each vector in $B$ also requires $O(d+mk|\Sigma|)$ time, and the total time cost producing $B$ is $O(dm|\Sigma|^k+m^2k|\Sigma|^{k+1})$. So the total time cost of this reduction is $O(dm^2k^2|\Sigma|(n+|\Sigma|^k))$.

\noindent\textbf{Proof of (i).}
Suppose there exist $\vec{v}_1\in V_1,\cdots,\vec{v}_k\in V_k$ satisfying $\sum_{i\in[k]}\vec{v}_i=\vec{t}$.
For every $i\in[k]$ we choose a vector $\vec{a}_{i,\vec{v}_i}\in A_i$.
And for every $j\in[m]$ we choose a vector $\vec{b}_{j,\vec{\sigma}_j}\in B_j$, where $\vec\sigma_j=(C(\vec v_1)[j],\dots,C(\vec v_m)[j])\in\Sigma^k$.
We now examine that $\sum_{i\in[k]}\vec a_{i,\vec v_i}+\sum_{j\in[m]}\vec b_{j,\vec\sigma_j}=\vec t'$ as:
\begin{itemize}
    \item For the first block, 
    \[
        \sum_{i\in[k]}\vec a_{i,\vec v_i}^{(1)}+\sum_{j\in[m]}\vec b_{j,\vec\sigma_j}^{(1)}=\sum_{i\in[k]}\vec v_i+\sum_{j\in[m]}\vec 0_{d}=\vec t=\vec t'^{(1)}.
    \]
    \item For every $j\in[m]$ the $(2,j)$-th block,
    \begin{align*}
        \sum_{i\in[k]}\vec a_{i,\vec v_i}^{(2,j)}+\sum_{j'\in[m]}\vec b_{j',\vec\sigma_{j'}}^{(2,j)}
        &=\sum_{i\in[k]}\vec a_{i,\vec v_i}^{(2,j)}+\vec b_{j,\vec\sigma_{j}}^{(2,j)}\\
        &=\sum_{i\in[k]}(\overbrace{\vec 0,\dots,\vec 0}^{i-1},\vec e_{C(\vec v_i)[j]},\vec 0,\dots,\vec 0)+(-\vec e_{C(\vec v_1)[j]},\dots,-\vec e_{C(\vec v_k)[j]})\\
        &=(\vec e_{C(\vec v_1)[j]},\dots,\vec e_{C(\vec v_k)[j]})+(-\vec e_{C(\vec v_1)[j]},\dots,-\vec e_{C(\vec v_k)[j]})\\
        &=\vec 0_{k|\Sigma|}=\vec t'^{(2,j)}.
    \end{align*}
        
    \item For the third block,
    \[
        \sum_{i\in[k]}\vec a_{i,\vec v_i}^{(3)}+\sum_{j\in[m]}\vec b_{j,\vec\sigma_j}^{(3)}=\sum_{i\in[k]}\vec e_i+\sum_{j\in[m]}\vec 0_{k}=\vec 1_k=\vec t'^{(3)}.
    \]
    \item For the fourth block,
    \[
        \sum_{i\in[k]}\vec a_{i,\vec v_i}^{(4)}+\sum_{j\in[m]}\vec b_{j,\vec\sigma_j}^{(4)}=\sum_{i\in[k]}\vec 0_m+\sum_{j\in[m]}\vec e_{j}=\vec 1_m=\vec t'^{(4)}.
    \]
\end{itemize}

\noindent\textbf{Proof of (ii).}
Suppose $X\subseteq A\dot\cup B$ and $\lambda:X\to\F_p^+$ such that $\sum_{\vec{x}\in X}\lambda(\vec x)\vec{x}=\vec{t}'$. Observe the third block of the equation:
\[
    \sum_{\vec x\in X}\lambda(\vec x)\vec x^{(3)}=\sum_{i\in[k]}\sum_{\vec x\in X\cap A_i}\lambda(\vec x)\vec e_i=\vec 1_{m}=\vec t'^{(3)}.
\]
For every $i\in[k]$, $X\cap A_i$ must not be empty since $\sum_{\vec x\in X\cap A_i}\lambda(\vec x)=1$. Also similarly by observing the fourth block it holds that $X\cap B_j$ must not be empty for every $j\in[m]$.
Therefore $|X\cap A|\geq k$ and $|X\cap B|\geq m$.

Further suppose that any $\vec v_1\in V_1,\dots,\vec v_k\in V_k$ and $\alpha_1,\dots\alpha_k\in\F_p^+$ must satisfy $\alpha_1\vec v_1+\dots+\alpha_k\vec v_k\neq\vec t$, we show that either $|X\cap A|\geq ck$ or $|X\cap B|\geq 2(1-\varepsilon)m$.

We let $I\subseteq[m]$ be the set of indices $j$ that $X\cap B_j$ contains only one vector, i.e.,
\[
    I=\{j\in[m] : |X\cap B_j|=1\}.
\]
Since $|X\cap B_j|\geq 1$ for every $j\in[m]$, if $|I|\leq\varepsilon m$ then 
\[
    |X\cap B|\geq\sum_{j\in[m]\backslash I}|X\cap B_j|\geq2(1-\varepsilon)m
\]
as desired. It remains to show that if $|I|>\varepsilon m$ then $|X\cap A|\geq ck$.

First we claim that there must be an $i\in[k]$ such that $X\cap A_i$ contains more than one vector. Otherwise suppose that $|X\cap A_i|=1$ for every $i\in[k]$, let $\vec a_{i,\vec v_i}\in X\cap A_i$ be the unique vector in $X\cap A_i$. Recall that in the first block, vectors in $X\cap B$ are all zero, so the sum of vectors in $X$ in the first block is
\[
    \sum_{\vec x\in X}\lambda(\vec x)\vec x^{(1)}=\sum\lambda(\vec a_{i,\vec v_i})\vec a_{i,\vec v_i}'^{(1)}=\sum_{i\in[k]}\lambda(\vec a_{i,\vec v_i})\vec v_i=\vec t=\vec t'^{(1)}
\]
This contradicts to our assumption that for all $\vec{v}_1\in V_1,\ldots,\vec{v}_k\in V_k$ and $\alpha_1,\dots,\alpha_k\in\F_p^+$, $\sum_{i\in[k]}\alpha_i\vec{v}_i\neq\vec{t}$. Therefore, there must be such an index $i^*\in[k]$ that $|A_{i^*}'|>1$.

Let $l>1$ be the size of $X\cap A_{i^*}$, we next show that $l\geq ck$. Suppose that $X\cap A_{i^*}=\{\vec a_{i^*,\vec v_1},\dots,\vec a_{i^*,\vec v_l}\}$ where $\vec v_1,\dots\vec v_l\in V_{i^*}$. We show in the following that the codeword set $\{C(\vec v_1),\dots,C(\vec v_l)\}$ must collide on every $j\in I$. Fix any $j\in I$, let $\vec b_{j,\vec \sigma}$ be the unique vector in $X\cap B_j$, where $\vec\sigma=(\sigma_1,\dots,\sigma_k)$. Recall that the $(2,j)$-th block of the resulting dimension consists of $k|\Sigma|$ coordinates, here we further break it down into $k$ sub-blocks each of size $|\Sigma|$, and we focus on the $(2,j,i^{*})$-th sub-block: 
\begin{align*}
    \sum_{\vec x\in X}\lambda(\vec x)\vec x^{(2,j,i^*)}
    &=\lambda(\vec a_{i^*,\vec v_1})\vec a_{i^*,\vec v_1}^{(2,j,i^*)}+\dots+\lambda(\vec a_{i^*,\vec v_l})\vec a_{i^*,\vec v_l}^{(2,j,i^*)}+\lambda(\vec b_{j,\vec\sigma})\vec b_{j,\vec\sigma}^{(2,j,i^*)}\\
    &=\lambda(\vec a_{i^*,\vec v_1})\vec e_{C(\vec v_1)[j]}+\dots+\lambda(\vec a_{i^*,\vec v_l})\vec e_{C(\vec v_l)[j]}-\lambda(\vec b_{j,\vec\sigma})\vec e_{\sigma_{i^*}}\\
    &=\vec 0_{|\Sigma|}=\vec t'^{(2,j,i^*)}.
\end{align*}
If $C(\vec v_1)[j],\dots,C(\vec v_l)[j]$ are all distinct, the equation $\lambda(\vec a_{i^*,\vec v_1})\vec e_{C(\vec v_1)[j]}+\dots+\lambda(\vec a_{i^*,\vec v_l})\vec e_{C(\vec v_l)[j]}-\lambda(\vec b_{j,\vec\sigma})\vec e_{\sigma_{i^*}}=\vec 0_{|\Sigma|}$ must not be satisfied since $l>1$ and the $\lambda$'s are nonzero. Therefore $\{C(\vec v_1),\dots,C(\vec v_l)\}$ must collide on the $j$-th coordinate.

If $|I|>\varepsilon m$ then $\{C(\vec v_1),\dots,C(\vec v_l)\}$ collide on more than $\varepsilon m$ coordinates, by the definition of collision number, it holds that $|\{C(\vec v_1),\dots,C(\vec v_l)\}|\geq \mathrm{Col}_\varepsilon(C)\geq ck$. And thus $|X\cap A|\geq|X\cap A_{i^*}|\geq ck$.
\end{proof}

Since the codes (with good collision number) we construct has codeword length $m=O(k^2\log k)$ much greater that $k$, the above construction cannot directly leads to a gap-creating reduction for \textsc{$k$-MLD}.
To settle this, intuitively we further duplicate the vector sets $A_1,\dots,A_k$ several times into $m$ vector sets. This leads to our gap creating reduction as follows.

\begin{theorem}\label{theorem: gap creating (main)}
For any $0<\varepsilon<1$, there is a randomized reduction which on input $k$ vector sets $V_1,\cdots,V_k\subseteq\mathbb{F}_p^d$ each of size $n$ and a target vector $\vec{t}\in\mathbb{F}_p^d$ outputs $k'$ vector sets $U_1,\dots,U_{k'}\subseteq\F_p^D$ and a target vector $\vec t'\in\F_p^D$ with $k'=O(k^2\log k)$ and $D=O(k'd+k'^2n^{1/k})$ in 
$O(d2^{O(k)}n^{1.01})$ time such that
\begin{itemize}
    \item[(i)]  If there exist $\vec{v}_1\in V_1,\ldots,\vec{v}_k\in V_k$ such that $\sum_{i\in[k]}\vec{v}_i=\vec{t}$, then there exists $\vec u_1\in U_1,\dots,\vec u_{k'}\in U_{k'}$ with their sum being $\vec{t}'$.
    \item[(ii)] If any $\vec v_1\in V_1,\dots,\vec v_k\in V_k$ and $\alpha_1,\dots\alpha_k\in\F_p^+$ must satisfy $\alpha_1\vec v_1+\dots+\alpha_k\vec v_k\neq\vec t$, then any $X\subseteq \bigcup_{i\in[k']}U_i$ and $\lambda:X\to\F_p^+$ such that $\sum_{\vec{x}\in X}\lambda(\vec x)\vec{x}=\vec{t}'$ must satisfy $|X|\geq(\frac{3}{2}-\varepsilon)k'$.
\end{itemize}
\end{theorem}
\begin{proof}
Suppose we are given input $V_1,\cdots,V_k\subseteq\mathbb{F}_p^d$ each of size $n$ and a target vector $\vec{t}\in\mathbb{F}_p^d$. We apply Lemma~\ref{lemma: collision number from random code (grouped)} with $c=2$. Then we obtain a code $C\subseteq \Sigma^{m}$ with $|C|=n,|\Sigma|=O(n^{1/k})$ and $m=O(k^2\log k)$ such that with high probability $\mathrm{Col}_\varepsilon(C)>2k$. Further combining with the construction of Lemma~\ref{lemma: gap creating (warm up)} we have the resulting vector sets $A_1,\dots,A_k,B_1,\dots,B_m\subseteq\F_p^D$ and $\vec t'\in\F_p^D$.
Let $w=m/k$ and we stretch the dimension $w$ times, i.e., our resulting dimension is $D'=wD$.
Our output are $k'=2m$ vector sets $(\bigcup_{l\in[w],i\in[k]}A_{l,i}') \cup (\bigcup_{j\in[m]}B_j')\subseteq \F_p^{D'}$ and $\vec t''\in\F_p^{D'}$ where
\[
    A_{l,i}'=\{(\underbrace{\overbrace{\vec 0_D,\dots,\vec 0_D}^{l-1},\vec a,\vec 0_D,\dots,\vec 0_D}_{w})\in\F_p^{D'} \mid \vec a\in A_i\} \text{ for every $l\in[w]$, $i\in[k]$},\\
\]
\[
    B_{j}'=\{(\underbrace{\vec b,\dots,\vec b}_{w})\in \F_p^{D'}\mid \vec b\in B_j\}\text{ for every $j\in[m]$},\\
\]
and
\[
    \vec t'' = (\underbrace{\vec t',\dots,\vec t'}_{w})\in \F_p^{D'}.
\]
Also for convenience we further define $A_{l}'=A_{l,1}'\cup\dots\cup A_{l,k}'$ for every $l\in[w]$ and define $A'=A_1'\cup\dots\cup A_w', B'=B_1'\cup\dots\cup B_m'$.

For time complexity, procedure in Lemma~\ref{lemma: gap creating (warm up)} requires $O(dm^2k^2|\Sigma|(n+|\Sigma|^k))$-time, and copying $w=m/k$ times requires $w$ times of time above. Hence the total time cost is $O(dm^3k|\Sigma|(n+|\Sigma|^k))$, or considering the parameters we choose, $O(dc^kk^7(\log k)^3n^{1+1/k})=O(d2^{O(k)}n^{1.01})$.

Suppose that there are $\vec a_1\in A_1,\dots,\vec a_k\in A_k$ and $\vec b_1\in B_1,\dots,\vec b_m\in B_m$ that sum to $\vec t'$.
Then for every $l\in[w],i\in[i]$ select $(\vec 0_{(l-1)D},\vec a_i,\vec 0_{(w-l)D})$ from $A_{l,i}'$ and for every $j\in[m]$ select $(\vec b_j,\dots,\vec b_j)$ from $B_{j}'$. One can see that these vectors have their sum being $(\vec t',\dots,\vec t')=\vec t''$ as desired.

Assume $X\subseteq (\bigcup_{l\in[w],i\in[k]}A_{l,i}') \cup (\bigcup_{j\in[m]}B_j')$ and $\lambda:X\to\F_p^+$ such that $\sum_{\vec{x}\in X}\lambda(\vec x)\vec{x}=\vec{t}''$. Fix any $l\in[w]$ and focus on the $l$-th block of the resulting dimension. On these coordinates, vector sets $A_{l,1}',\dots,A_{l,k}',B_1',\dots,B_m'$ and $\vec t''$ plays exactly the same role as $A_1,\dots,A_k,B_1,\dots,B_m$ and $\vec t'$, and all remaining vectors have zero entries. As also in Lemma~\ref{lemma: gap creating (warm up)}, it holds that $|X'\cap A_l'|\geq k$ and $|X'\cap B'|\geq m$.

Now we further suppose that any $\vec v_1\in V_1,\dots,\vec v_k\in V_k$ and $\alpha_1,\dots\alpha_k\in\F_p^+$ must satisfy $\alpha_1\vec v_1+\dots+\alpha_k\vec v_k\neq\vec t$.
First consider the case that $|X'\cap B'|<2(1-\varepsilon)m$. Then due to the property of Lemma~\ref{lemma: gap creating (warm up)}, for every $l\in[w]$, $|X'\cap A_l'|\geq 2k$. Therefore in this case, $|X'|=|X'\cap A'|+|X'\cap B'|\geq w\cdot(2k)+m=3m=\frac{3}{2}k'$.
For the other case that $|X'\cap B'|\geq 2(1-\varepsilon)m$, we have $|X'|=|X'\cap A'|+|X'\cap B'|\geq w\cdot k+2(1-\varepsilon)m=(\frac{3}{2}-\varepsilon)k'$. Therefore in both cases, $|X'|\geq (\frac{3}{2}-\varepsilon)k'$ as desired.
\end{proof}


\begin{remark}
    Consider the \textsc{$k$-VectorSum$_q$} problem in \cite{LinRSW22}, whose definition is identical to \textsc{$k$-MLD$_q$} except that it requires all the coefficients being $1$. A closer look at our reduction shows that it can directly create a gap of almost $(q+1)/2$ for \textsc{$k$-VectorSum$_q$} rather than almost $\frac{3}{2}$ in the \textsc{$k$-MLD$_q$} case. The reason is that when coefficients are fixed to $1$, for any solution $X$ and each $j\in[m]$ with $|X\cap B_j|>1$, it must satisfies $|X\cap B_j|=cq+1$ for some positive integer $c$ so that the final block of vectors in $X\cap B_j$ can have sum $\vec{e}_j$. Thus if some solution having less than $\varepsilon$ fraction of $j\in[m]$ with $|X\cap B_j|=1$, It must satisfies $|X\cap B|\geq q(1-\varepsilon)m$ instead of $2(1-\varepsilon)m$ in \textsc{$k$-MLD$_q$} case, and the final approximation ratio can be improved to $(\frac{q+1}{2}-\varepsilon)$, significantly larger than $(\frac{3}{2}-\varepsilon)$ when $q$ is superconstant. 
\end{remark}
\section{Lower Bounds for Gap-$k$-NCP and Other Problems}\label{section: from NCP}
In this section, we show the reduction described in the previous sections implies improved running time lower bounds for various problems under the Exponential Time Hypothesis (ETH). 
\subsection{Maximum Likelihood Decoding and Nearest Codeword Problem}

Bhattacharyya, Ghoshal, Karthik and Manurangsi~\cite{BGKM18} presented a gap amplification procedure for \textsc{Gap-$k$-MLD$_p$}. Although they only discussed the procedure on the binary field, it's straightforward to see the procedure also works for \textsc{Gap-$k$-MLD$_p$} instances over all $\mathbb{F}_p$. Formally,

\begin{theorem}[Generalization of Lemma 4.5 in \cite{BGKM18}]
\label{theorem: gap amplification for mld in BGKM18}
    For integers $k_1,k_2>0$, $k'=k_2+k_1k_2$ and reals $\gamma_1,\gamma_2>1$, $\gamma'\geq \gamma_1\gamma_2(1-\frac{1}{k_1})$, there is a polynomial time algorithm that on input $2$ vector sets $U\subseteq \mathbb{F}^{m_1}_p, V\subseteq \mathbb{F}^{m_2}_p$, $|U|=n_1,|V|=n_2$, two target vectors $\Vec{t}\in\mathbb{F}^{m_1}_p,\vec{s}\in\mathbb{F}^{m_2}_p$, outputs a vector set $W\subseteq \mathbb{F}^{m_2+n_1 m_1}_p$ and a target vector $\Vec{t}'\in \mathbb{F}^{m_2+n_1 m_1}_p$ satisfies:
    \begin{itemize}
        \item If $(U,\vec{t})$ is a YES instance of \textsc{$\gamma_1$-Gap-$k_1$-MLD$_p$} instance and $(V,\vec{s})$ is a YES instance of \textsc{$\gamma_2$-Gap-$k_2$-MLD$_p$} instance, then $(W,\vec{t}')$ is a YES instance of \textsc{$\gamma'$-Gap-$k'$-MLD$_p$}.
        \item If $(U,\vec{t})$ is a NO instance of \textsc{$\gamma_1$-Gap-$k_1$-MLD$_p$} instance and $(V,\vec{s})$ is a NO instance of \textsc{$\gamma_2$-Gap-$k_2$-MLD$_p$} instance, then $(W,\vec{t}')$ is a NO instance of \textsc{$\gamma'$-Gap-$k'$-MLD$_p$}.
    \end{itemize}
\end{theorem}

Readers seeking for a formal proof is referred to \cite[Section 4.2]{BGKM18}, we only present a figure showing their construction in a intuitive way in Figure~\ref{figure for gap amplification in BGKM18}.

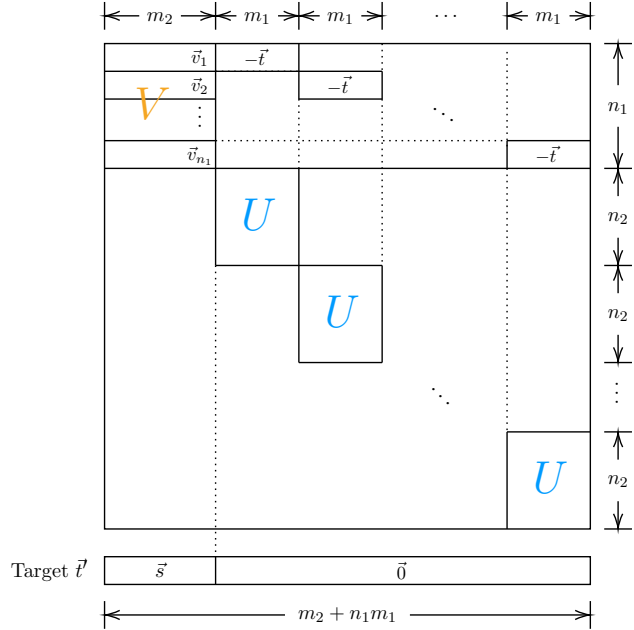
\begin{figure}[h]
\centering
\scalebox{0.7}{
\tikzset{every picture/.style={line width=0.75pt}} 

\begin{tikzpicture}[x=0.75pt,y=0.75pt,yscale=-1,xscale=1]

\draw (100,-10) -- (100,10) ;
\draw (180,-10) -- (180,10) ;
\draw (240,-10) -- (240,10) ;
\draw (300,-10) -- (300,10) ;
\draw (390,-10) -- (390,10) ;
\draw (450,-10) -- (450,10) ;

\draw    (100,0) -- (125,0) ;
\draw [shift={(100,0)}, rotate = 0] [color={rgb, 255:red, 0; green, 0; blue, 0 }  ][line width=0.75]    (10.93,-3.29) .. controls (6.95,-1.4) and (3.31,-0.3) .. (0,0) .. controls (3.31,0.3) and (6.95,1.4) .. (10.93,3.29)   ;
\draw    (155,0) -- (180,0) ;
\draw [shift={(180,0)}, rotate = 180] [color={rgb, 255:red, 0; green, 0; blue, 0 }  ][line width=0.75]    (10.93,-3.29) .. controls (6.95,-1.4) and (3.31,-0.3) .. (0,0) .. controls (3.31,0.3) and (6.95,1.4) .. (10.93,3.29)   ;
\draw (129,-4) node [anchor=north west][inner sep=0.75pt]   [align=left] {$\displaystyle m_{2}$};

\draw    (180,0) -- (195,0) ;
\draw [shift={(180,0)}, rotate = 0] [color={rgb, 255:red, 0; green, 0; blue, 0 }  ][line width=0.75]    (10.93,-3.29) .. controls (6.95,-1.4) and (3.31,-0.3) .. (0,0) .. controls (3.31,0.3) and (6.95,1.4) .. (10.93,3.29)   ;
\draw    (225,0) -- (240,0) ;
\draw [shift={(240,0)}, rotate = 180] [color={rgb, 255:red, 0; green, 0; blue, 0 }  ][line width=0.75]    (10.93,-3.29) .. controls (6.95,-1.4) and (3.31,-0.3) .. (0,0) .. controls (3.31,0.3) and (6.95,1.4) .. (10.93,3.29)   ;
\draw (199,-4) node [anchor=north west][inner sep=0.75pt]   [align=left] {$\displaystyle m_{1}$};

\draw    (240,0) -- (255,0) ;
\draw [shift={(240,0)}, rotate = 0] [color={rgb, 255:red, 0; green, 0; blue, 0 }  ][line width=0.75]    (10.93,-3.29) .. controls (6.95,-1.4) and (3.31,-0.3) .. (0,0) .. controls (3.31,0.3) and (6.95,1.4) .. (10.93,3.29)   ;
\draw    (285,0) -- (300,0) ;
\draw [shift={(300,0)}, rotate = 180] [color={rgb, 255:red, 0; green, 0; blue, 0 }  ][line width=0.75]    (10.93,-3.29) .. controls (6.95,-1.4) and (3.31,-0.3) .. (0,0) .. controls (3.31,0.3) and (6.95,1.4) .. (10.93,3.29)   ;
\draw (259,-4) node [anchor=north west][inner sep=0.75pt]   [align=left] {$\displaystyle m_{1}$};

\draw (336,-4) node [anchor=north west][inner sep=0.75pt]   [align=left] {$\displaystyle \cdots$};

\draw    (390,0) -- (405,0) ;
\draw [shift={(390,0)}, rotate = 0] [color={rgb, 255:red, 0; green, 0; blue, 0 }  ][line width=0.75]    (10.93,-3.29) .. controls (6.95,-1.4) and (3.31,-0.3) .. (0,0) .. controls (3.31,0.3) and (6.95,1.4) .. (10.93,3.29)   ;
\draw    (435,0) -- (450,0) ;
\draw [shift={(450,0)}, rotate = 180] [color={rgb, 255:red, 0; green, 0; blue, 0 }  ][line width=0.75]    (10.93,-3.29) .. controls (6.95,-1.4) and (3.31,-0.3) .. (0,0) .. controls (3.31,0.3) and (6.95,1.4) .. (10.93,3.29)   ;
\draw (409,-4) node [anchor=north west][inner sep=0.75pt]   [align=left] {$\displaystyle m_{1}$};

\draw   (100,20) -- (450,20) -- (450,370) -- (100,370) -- cycle ;
\draw    (100,110) -- (450,110) ;
\draw    (180,20) -- (180,110) ;
\draw  [dash pattern={on 0.84pt off 2.51pt}]  (180,180) -- (180,390) ;
\draw    (180,180) -- (300,180) ;
\draw    (180,110) -- (180,180) ;
\draw    (240,20) -- (240,60) ;
\draw    (100,40) -- (240,40) ;
\draw    (240,60) -- (300,60) ;
\draw    (240,40) -- (300,40) ;
\draw    (300,40) -- (300,60) ;
\draw  [dash pattern={on 0.84pt off 2.51pt}]  (240,60) -- (240,110) ;
\draw  [dash pattern={on 0.84pt off 2.51pt}]  (300,20) -- (300,40) ;
\draw    (240,250) -- (300,250) ;
\draw    (300,180) -- (300,250) ;
\draw    (240,110) -- (240,250) ;
\draw  [dash pattern={on 0.84pt off 2.51pt}]  (300,60) -- (300,180) ;
\draw    (100,60) -- (180,60) ;
\draw  [dash pattern={on 0.84pt off 2.51pt}]  (180,40) -- (240,40) ;
\draw    (100,90) -- (180,90) ;
\draw  [dash pattern={on 0.84pt off 2.51pt}]  (180,90) -- (390,90) ;
\draw    (390,90) -- (450,90) ;
\draw  [dash pattern={on 0.84pt off 2.51pt}]  (390,26) -- (390,96) ;
\draw    (390,90) -- (390,110) ;
\draw    (390,300) -- (450,300) ;
\draw    (390,300) -- (390,370) ;
\draw  [dash pattern={on 0.84pt off 2.51pt}]  (390,110) -- (390,300) ;
\draw   (100,390) -- (450,390) -- (450,410) -- (100,410) -- cycle ;
\draw    (180,390) -- (180,410) ;
\draw    (100,422) -- (100,442) ;
\draw    (450,422) -- (450,442) ;
\draw    (320,432) -- (448,432) ;
\draw [shift={(450,432)}, rotate = 180] [color={rgb, 255:red, 0; green, 0; blue, 0 }  ][line width=0.75]    (10.93,-3.29) .. controls (6.95,-1.4) and (3.31,-0.3) .. (0,0) .. controls (3.31,0.3) and (6.95,1.4) .. (10.93,3.29)   ;
\draw    (230,432) -- (102,432) ;
\draw [shift={(100,432)}, rotate = 360] [color={rgb, 255:red, 0; green, 0; blue, 0 }  ][line width=0.75]    (10.93,-3.29) .. controls (6.95,-1.4) and (3.31,-0.3) .. (0,0) .. controls (3.31,0.3) and (6.95,1.4) .. (10.93,3.29)   ;
\draw    (460,20) -- (480,20) ;
\draw    (460,110) -- (480,110) ;
\draw    (470,80) -- (470,108) ;
\draw [shift={(470,110)}, rotate = 270] [color={rgb, 255:red, 0; green, 0; blue, 0 }  ][line width=0.75]    (10.93,-3.29) .. controls (6.95,-1.4) and (3.31,-0.3) .. (0,0) .. controls (3.31,0.3) and (6.95,1.4) .. (10.93,3.29)   ;
\draw    (470,50) -- (470,22) ;
\draw [shift={(470,20)}, rotate = 90] [color={rgb, 255:red, 0; green, 0; blue, 0 }  ][line width=0.75]    (10.93,-3.29) .. controls (6.95,-1.4) and (3.31,-0.3) .. (0,0) .. controls (3.31,0.3) and (6.95,1.4) .. (10.93,3.29)   ;
\draw    (460,180) -- (480,180) ;
\draw    (470,160) -- (470,178) ;
\draw [shift={(470,180)}, rotate = 270] [color={rgb, 255:red, 0; green, 0; blue, 0 }  ][line width=0.75]    (10.93,-3.29) .. controls (6.95,-1.4) and (3.31,-0.3) .. (0,0) .. controls (3.31,0.3) and (6.95,1.4) .. (10.93,3.29)   ;
\draw    (470,130) -- (470,112) ;
\draw [shift={(470,110)}, rotate = 90] [color={rgb, 255:red, 0; green, 0; blue, 0 }  ][line width=0.75]    (10.93,-3.29) .. controls (6.95,-1.4) and (3.31,-0.3) .. (0,0) .. controls (3.31,0.3) and (6.95,1.4) .. (10.93,3.29)   ;
\draw    (460,250) -- (480,250) ;
\draw    (470,230) -- (470,248) ;
\draw [shift={(470,250)}, rotate = 270] [color={rgb, 255:red, 0; green, 0; blue, 0 }  ][line width=0.75]    (10.93,-3.29) .. controls (6.95,-1.4) and (3.31,-0.3) .. (0,0) .. controls (3.31,0.3) and (6.95,1.4) .. (10.93,3.29)   ;
\draw    (470,200) -- (470,182) ;
\draw [shift={(470,180)}, rotate = 90] [color={rgb, 255:red, 0; green, 0; blue, 0 }  ][line width=0.75]    (10.93,-3.29) .. controls (6.95,-1.4) and (3.31,-0.3) .. (0,0) .. controls (3.31,0.3) and (6.95,1.4) .. (10.93,3.29)   ;
\draw    (460,370) -- (480,370) ;
\draw    (470,350) -- (470,368) ;
\draw [shift={(470,370)}, rotate = 270] [color={rgb, 255:red, 0; green, 0; blue, 0 }  ][line width=0.75]    (10.93,-3.29) .. controls (6.95,-1.4) and (3.31,-0.3) .. (0,0) .. controls (3.31,0.3) and (6.95,1.4) .. (10.93,3.29)   ;
\draw    (470,320) -- (470,302) ;
\draw [shift={(470,300)}, rotate = 90] [color={rgb, 255:red, 0; green, 0; blue, 0 }  ][line width=0.75]    (10.93,-3.29) .. controls (6.95,-1.4) and (3.31,-0.3) .. (0,0) .. controls (3.31,0.3) and (6.95,1.4) .. (10.93,3.29)   ;
\draw    (460,300) -- (480,300) ;

\draw (120,50) node [anchor=north west][inner sep=0.75pt]  [font=\fontsize{3.67em}{4.4em}\selectfont,color={rgb, 255:red, 245; green, 166; blue, 35 }  ,opacity=1 ] [align=left] {$\displaystyle V$};
\draw (197,130) node [anchor=north west][inner sep=0.75pt]  [font=\Huge,color={rgb, 255:red, 0; green, 155; blue, 255 }  ,opacity=1 ] [align=left] {$\displaystyle U$};
\draw (257,200) node [anchor=north west][inner sep=0.75pt]  [font=\Huge,color={rgb, 255:red, 0; green, 155; blue, 255 }  ,opacity=1 ] [align=left] {$\displaystyle U$};
\draw (407,320) node [anchor=north west][inner sep=0.75pt]  [font=\Huge,color={rgb, 255:red, 0; green, 155; blue, 255 }  ,opacity=1 ] [align=left] {$\displaystyle U$};
\draw (199,22) node [anchor=north west][inner sep=0.75pt] [font=\small]  [align=left] {$\displaystyle -\vec{t}$};
\draw (259,42) node [anchor=north west][inner sep=0.75pt] [font=\small]  [align=left] {$\displaystyle -\vec{t}$};
\draw (409,92) node [anchor=north west][inner sep=0.75pt]  [font=\small] [align=left] {$\displaystyle -\vec{t}$};
\draw (135,393) node [anchor=north west][inner sep=0.75pt]   [align=left] {$\displaystyle \vec{s}$};
\draw (309,392) node [anchor=north west][inner sep=0.75pt]  [font=\small] [align=left] {$\displaystyle \vec{0}$};
\draw (31,390) node [anchor=north west][inner sep=0.75pt]   [align=left] {Target $\displaystyle \vec{t}'$};
\draw (161,23) node [anchor=north west][inner sep=0.75pt]  [font=\small] [align=left] {$\displaystyle \vec{v}_{1}$};
\draw (161,43) node [anchor=north west][inner sep=0.75pt]  [font=\small] [align=left] {$\displaystyle \vec{v}_{2}$};
\draw (158,93) node [anchor=north west][inner sep=0.75pt]  [font=\small] [align=left] {$\displaystyle \vec{v}_{n_{1}}$};
\draw (173,63) node [anchor=north west][inner sep=0.75pt]  [rotate=-90] [align=left] {$\displaystyle \cdots $};
\draw (337.55,61.22) node [anchor=north west][inner sep=0.75pt]  [font=\large,rotate=-32.98] [align=left] {$\displaystyle \cdots $};
\draw (337.55,263.45) node [anchor=north west][inner sep=0.75pt]  [font=\large,rotate=-43.81] [align=left] {$\displaystyle \cdots $};
\draw (238,425) node [anchor=north west][inner sep=0.75pt]   [align=left] {$\displaystyle m_{2} +n_{1} m_{1}$};
\draw (461,61) node [anchor=north west][inner sep=0.75pt]   [align=left] {$\displaystyle n_{1}$};
\draw (461,140) node [anchor=north west][inner sep=0.75pt]   [align=left] {$\displaystyle n_{2}$};
\draw (461,210) node [anchor=north west][inner sep=0.75pt]   [align=left] {$\displaystyle n_{2}$};
\draw (461,330) node [anchor=north west][inner sep=0.75pt]   [align=left] {$\displaystyle n_{2}$};
\draw (473.5,262) node [anchor=north west][inner sep=0.75pt]  [rotate=-90] [align=left] {$\displaystyle \cdots $};

\end{tikzpicture}
}
\caption{A pictorial illustration for the construction in Theorem \ref{theorem: gap amplification for mld in BGKM18}.}
    \label{figure for gap amplification in BGKM18}
\end{figure}

\subsubsection{ETH-based Running Time Lower Bound}
Taking a closer look at the reduction from \textsc{$3$-SAT} to \textsc{$k$-VectorSum} in \cite[Theorem 11]{LinRSW22}, we observe that by applying a minor modification, their reduction can actually have soundness condition as:
\begin{itemize}
    \item If $\phi$ is not satisfiable, then for any $\Vec{v}_1\in V_1, \cdots, \Vec{v}_k\in V_k$ and $\alpha_1,\cdots, \alpha_k \in \mathbb{F}_p^+$, $\Sigma_{i=1}^k \alpha_i \Vec{v}_i \neq \vec{t}$.
\end{itemize}

The modification is simply appending a vector $(0^{i-1}\circ 1 \circ 0^{k-i})$ to each vector in $V_i$, for all $1\leq i\leq k$. Then, the target vector is changed from a zero vector to $\vec{t}=0^d\circ 1^k$.
Completeness of their reduction is trivially preserved. For soundness we claim, we note that for any $\Vec{v}_1\in V_1, \cdots, \Vec{v}_k\in V_k$ and $\alpha_1,\cdots, \alpha_k \in \mathbb{F}_p$, if $\Sigma_{i=1}^k \alpha_i \Vec{v}_i = \vec{t}$, then $\alpha_1 = \cdots = \alpha_k = 1$. 

By strengthening the soundness condition in \cite{LinRSW22}, we obtain exactly the restricted version of \textsc{$k$-MLD$_p$} in the previous sections. Combining with their soundness for \textsc{$k$-VectorSum}, we obtain the following hardness result for \textsc{$k\text{-MLD}_p$} as:

\begin{theorem}[Theorem 11 in~\cite{LinRSW22}]\label{theorem: hardness of MLD from ETH}
    Assuming ETH, for any constant integer $p$, \textsc{$k$-MLD$_p$} has no $n^{o(k)}$-time algorithm.
\end{theorem}

The parameterized \textsc{MLD} and \textsc{NCP} are equivalent in the sense that there exists reductions preserving the solution size in both direction, see Appendix \ref{appendix: equivalence between MLD and NCP}. Recall that Theorem \ref{theorem: gap creating (main)} showed a reduction from \textsc{$k$-MLD$_p$} to \textsc{$(3/2-\varepsilon)$-Gap-$k'$-MLD$_p$} with $k'=k^2\log k$ and $\varepsilon>0$. Combining running time lower bound in Theorem \ref{theorem: hardness of MLD from ETH}, we have:
\begin{theorem}\label{thm: 3/2-gap mld ETH lower bound}
    Assuming randomized ETH, for any constant integer $p$, constant $1<\gamma<\frac{3}{2}$, \textsc{$\gamma$-Gap-$k$-MLD$_p$} and \textsc{$\gamma$-Gap-$k$-NCP$_p$} has no $O_k(n^{o(\sqrt{k/\log k})})$-time algorithm.
\end{theorem}

By applying Theorem \ref{theorem: gap amplification for mld in BGKM18} to the gap instance itself $O(\log\log \gamma)$ times, we can obtain the ETH-based time lower bound for approximating parameterized \textsc{MLD} and \textsc{NCP} to any constant factor.

\begin{corollary}
    Assuming ETH, for any constant integer $p$ and constant $\gamma>1$, \textsc{$\gamma$-Gap-$k$-MLD$_p$} and \textsc{$\gamma$-Gap-$k$-NCP$_p$} have no $O_k(n^{o(k^{\epsilon})})$ time algorithm, where $\epsilon = \frac{1}{\textsf{polylog}(\gamma)}$ is a constant.
\end{corollary}

\subsection{Minimum Distance Problem}
The reduction from \textsc{Gap-$k$-NCP} to \textsc{Gap-$k$-MDP} in~\cite{BCGR23} is as follows.
\begin{theorem}[\cite{BCGR23}, Theorem 3.1 and 3.3]\label{thm:NCP to MDP in BCGR23}
    For any prime power $p\geq 2$ there is a randomized reduction from \textsc{$(4p)$-Gap-$k$-NCP$_p$} to \textsc{$\frac{4p}{4p-1}$-Gap-$k'$-MDP$_p$} runs in polynomial time with $k'=O(k)$.
\end{theorem}

Combining with our reduction for \textsc{Gap-$k$-MLD} and \textsc{Gap-$k$-NCP}, we have:
\begin{corollary}
    Assuming randomized ETH, for any prime power $p\geq 2$ and real number $\gamma>1$, \textsc{$\gamma$-Gap-$k$-MDP$_{p}$} has no $O_k(n^{o(k^\epsilon)})$ time algorithm, where $\epsilon=\Theta(\frac{1}{p\log \gamma\cdot\mathsf{polylog} (p)})$.
\end{corollary}
\begin{proof}
    We first apply Theorem \ref{theorem: gap creating (main)} to obtain a \textsc{$(\frac{3}{2}-\varepsilon)$-Gap-$k_1$-MLD$_p$}, where $k_1=k^2\log k$. Then apply Theorem \ref{theorem: gap amplification for mld in BGKM18} for $\Theta(\log\log p)$ times to amplify the gap to $4p$, this will cause the parameter $k_2$ grows to $k_1^{\textsf{polylog} (p)}=k^{\textsf{polylog} (p)}$. Reduction from \textsc{Gap-$k$-MLD} to \textsc{Gap-$k$-NCP} is trivial as in Theorem \ref{thm:MLDtoNCP}, and it preserves the parameter. Apply Theorem \ref{thm:NCP to MDP in BCGR23}, we obtain a \textsc{$\frac{4p}{4p-1}$-Gap-$k_3$-MDP$_p$} instance with $k_3=O(k_2)=k^{\textsf{polylog} (p)}$. Finally, to obtain any constant factor $\gamma>1$, it suffices to first self-tensor the instance for $O(\log p)$ times to obtain a \textsc{$2$-Gap} instance, causing a parameter growth of $k_4=k_3^{O(p)} = k^{p\cdot\textsf{polylog} (p)}$, then self-tensor it for $O(\log \log \gamma)$ times, with parameter growth $k_5=k_4^{O(\log \gamma)}=k^{p\log\gamma\cdot\textsf{polylog} (p)}$.
    
    Theorem \ref{theorem: hardness of MLD from ETH} showed an $n^{\Omega(k)}$ time lower bound for \textsc{$k$-MLD$_p$}, combining with our reduction, we have that under ETH, there are no $f(k)n^{o(k^\epsilon)}$ time algorithm for \textsc{$\gamma$-Gap-$k$-MDP$_p$} where $\epsilon = \Theta(\frac{1}{p\log \gamma\cdot\textsf{polylog} (p)})$.
\end{proof}

\subsection{Closest Vector Problem}
We need a reduction from \textsc{$(2\gamma)$-Gap-$k$-MLD$_q$} to \textsc{$\gamma$-Gap-$2k$-CVP$_p$} from  \cite{BBE+21}.
\begin{theorem}[\cite{BBE+21}, Theorem 7.2]\label{thm:MLD to CVP in BBE+21}
    For any real numbers $\gamma,p\geq 1$ and a prime number $q>2\gamma$, there is a reduction from \textsc{$(2\gamma)$-Gap-$k$-MLD$_q$} to \textsc{$\gamma$-Gap-$k'$-CVP$_p$} runs in polynomial time, where $k'=2k$.
\end{theorem}
For running time lower bound, we shall again analyze the parameter growth as follows.

\begin{corollary}\label{col:CVP lower bound under ETH}
    Assuming ETH, there exists constant $c>0$, for any real numbers $p,\gamma\geq 1$, \textsc{$\gamma$-Gap-$k$-CVP$_{p}$} has no $O_k(n^{o(k^\epsilon)})$ time algorithm, where $\epsilon=\Theta(\frac{1}{\gamma^c})$.
\end{corollary}
\begin{proof}
    We first apply Theorem \ref{theorem: gap creating (main)} to obtain a \textsc{$(\frac{3}{2}-\varepsilon)$-Gap-$k_1$-MLD$_p$}, where $k_1=k^2\log k$. Then apply Theorem \ref{theorem: gap amplification for mld in BGKM18} for $\Theta(\log\log \gamma)$ times to amplify the gap to $2\gamma$, this will cause the parameter $k_2$ grows to $k_1^{\textsf{polylog}(\gamma)}=k^{\textsf{polylog}(\gamma)}$. Then, apply Theorem \ref{thm:MLD to CVP in BBE+21}, we obtain a instance of \textsc{$\gamma$-Gap-$2k_2$-CVP$_p$}. Combining with Theorem \ref{theorem: hardness of MLD from ETH}, we obtain the lower bound of $f(k)n^{\Omega(k^\epsilon)}$, where $\epsilon= \Theta(\frac{1}{(\log\gamma)^c})$ and $c$ is a fixed constant independent of $\gamma$ and $p$.
\end{proof}

\subsection{Shortest Vector Problem}
Combining our work with \cite{BCGR23}, we show two ways of obtaining running time lower bound for \textsc{$\gamma$-Gap-$k$-SVP$_p$}. The first way reduces from \textsc{Gap-$k$-CVP$_p$}, obtaining lower bound for only a fixed constant ratio and all $l_p$ norms where $p\geq 1$. The second way reduces from \textsc{Gap-$k$-NCP$_q$}, obtaining lower bound for all constant ratio and all $l_p$ norms except for $l_1$.

\subsubsection{Reduction From \textsc{Gap-$k$-CVP$_p$}}
\begin{theorem}[\cite{BCGR23}, Theorem 4.1 and 4.3, modified]\label{thm:CVP to SVP}
    For any real numbers $p\geq 1$ and $\gamma'\in[1,2)$ there exist a real number $\gamma\geq 1$
    \footnote{$\gamma=(\max\left(12/\varepsilon,\frac{1}{(1+\varepsilon/2)^{1/p}-1}\right))^p$ where $\varepsilon=(\gamma')^{-1}-1/2>0$.}
    and a reduction from \textsc{$\gamma$-Gap-$k$-CVP$_p$} to \textsc{$\gamma'$-Gap-$k'$-SVP$_p$} runs in polynomial time, where $k'\leq \gamma k$.
\end{theorem}
\begin{corollary}
    Assuming randomized ETH, for any real numbers $p\geq 1$ and $\gamma\in[1,2)$, \textsc{$\gamma$-Gap-$k$-SVP$_{p}$} has no $O_k(n^{o(k^\epsilon)})$ time algorithm,where $0<\epsilon<1$ is some constant that depends on $p$ and $\gamma$.
\end{corollary}
\begin{proof}
    We have shown in Corollary \ref{col:CVP lower bound under ETH} the lower bound of \textsc{$\gamma_0$-Gap-$k$-CVP$_p$} to be $f(k)n^{\Omega(k^{\epsilon_0})}$, where $\epsilon_0=\Theta(\frac{1}{\gamma_0^c})$ with $c>0$ being a global constant. Set $\gamma_0$ to fit the constant in Theorem \ref{thm:CVP to SVP}, which also depends on only $p$ and $\gamma$, we obtain a reduction to \textsc{$\gamma$-Gap-$k'$-SVP$_{p}$} where $k'\leq \gamma k=O(k)$. This gives the lower bound of \textsc{$\gamma$-Gap-$k$-SVP$_{p}$} to be $f(k)n^{o(k^{\epsilon})}$ under ETH, for some constant $\epsilon>0$ depends only on $p$ and $\gamma$.
\end{proof}

\subsubsection{Reduction From \textsc{Gap-$k$-NCP$_2$}}

\begin{theorem}[\cite{BCGR23}, Lemma 5.1 and Theorem 5.2, modified]\label{thm:NCP to SVP}
    There exists a constant real $\mu\geq 1$ such that, for any real numbers $p> 1$ and $\gamma'\geq 1$, 
    there exists a reduction from \textsc{$\mu$-Gap-$k$-NCP$_2$} to \textsc{$\gamma'$-Gap-$k'$-SVP$_p$} runs in polynomial time, where $k'=O(k^{c})$, $c>1$ is a constant only depends on $p$ and $\gamma'$\footnote{
    There are two problems here about the parameter blow-up, one is that $k'\leq (\mu k)^{O(1)}$ due to the Haviv-Regev ``tensoring'' step of \textsc{SVP}, the other is that 
    to achieve final gap $\gamma'$, the gap $\mu$ of \textsc{NCP} needs to satisfy $\frac{\mu}{2^p+1+\alpha \mu}>\gamma'$ for some $1/2+2^{-p}<\alpha<1$, causing a polynomial blow-up of parameter to achieve such $\mu$.}.
\end{theorem}
The reduction in Theorem \ref{thm:NCP to SVP} in fact proceeds in two steps: first reduces \textsc{$\mu$-Gap-$k$-NCP$_2$} to \textsc{$\gamma'$-Gap-$k'$-SVP$_p$} for some fixed $\gamma'>1$ with $k'<\mu k$ (while having some additional properties for the second step), then use a tensor technique to amplify the gap to any constant.

\begin{corollary}
    Assuming randomized ETH, for any real numbers $p>1$ and $\gamma\geq 1$, \textsc{$\gamma$-Gap-$k$-SVP$_p$} has no $O_k(n^{o(k^{\epsilon})})$ time algorithm, where $0<\epsilon<1$ is some constant that depends on $p$ and $\gamma$.
\end{corollary}
\begin{proof}
     To fit the parameter requirement in the first step, we need a \textsc{$\mu$-Gap-$k_1$-NCP$_2$} instance, where $\mu$ is same as Theorem \ref{thm:NCP to SVP}. Such instances can be reduced from \textsc{$k$-MLD$_2$} with $k_1=O(k^{\epsilon_0})$ where $\epsilon_0$ is a constant depends only on $\mu$. Then, by applying Theorem \ref{thm:NCP to SVP}, we reduce \textsc{$\mu$-Gap-$k_1$-NCP$_2$} to \textsc{$\gamma$-Gap-$k_2$-SVP$_p$} for any $\gamma\geq 1$, and $k_2=O(k_1^{c})=O(k^{1/\epsilon})$, where $\epsilon=\Theta(\frac{1}{c})$ is a constant only depends on $p$ and $\gamma$ (and $\mu$, but omitted since it is a global constant indenpent with $p$ and $\gamma$). Therefore, under ETH, no algorithm can solve \textsc{$\gamma$-Gap-$k$-SVP$_p$} in time $f(k)n^{o(k^{\epsilon})}$.
\end{proof}
\section{Conclusion}
\label{Section: Conclusion}
We have presented new ETH-based lower bounds for approximating parameterized nearest codeword problem and its related problems, improving upon the previous results
from~\cite{BBE+21,BCGR23}. Our reduction technique is also simpler and more straight forward than the one used in~\cite{BBE+21}. However, our results still do not match the lower bound for constant Gap-$k$-NCP based on Gap-ETH~\cite{Man20}. A natural open problem is to close this gap by proving a stronger lower bound under an assumption that is weaker than Gap-ETH, such as constant Gap-$k$-Clique has no $n^{o(k)}$-time algorithm. This would be a key step towards understanding the fine-grained complexity of parameterized nearest codeword problem and its variants.
\begin{openproblem}
    Prove $n^{o(k)}$ time lower bound of approximating \textsc{$k$-NCP$_p$} or its related problems to any constant factor under assumptions weaker than Gap-ETH.
\end{openproblem}
To show such a result, as the comments in~\cite{Man20}, one might need to come up with a better “one-shot proof” that gives arbitrary constant factors without tensoring, and with linear parameter growth.

In this paper, we give a new method of composing \textit{threshold graph} with vector problems to yield hardness of approximation results. We showed the limitation of analyzing collision number of a code from its relative distance in~\cite{KN21,LRSW23}, and improved the analysis to bypass the limitation above. It might be interesting to consider whether this result can be further improved to yield threshold graph with better parameters, or some limitations of our method can be discovered, formally:
\begin{openproblem}
    Give a better construction of strong threshold graph in Section~\ref{Introduction-Technical Overview-Gap Creation} with $h=\Omega(k)$ and $m=O(k)$,     or show that such graphs do not exist.
\end{openproblem}

\bibliographystyle{alpha}
\bibliography{main}

\appendix
\section{Gap-preserving Reductions Between Parameterized MLD and NCP}\label{appendix: equivalence between MLD and NCP}
We first present a direct reduction from \textsc{$\gamma$-Gap-$k$-MLD$_{p}$} to \textsc{$\gamma$-Gap-$k$-NCP$_{p}$}.
\begin{theorem}
\label{thm:MLDtoNCP}
There is a deterministic FPT reduction that on input a \textsc{$\gamma$-Gap-$k$-MLD$_{p}$} instance with vector number $n$ and length $l$, output a \textsc{$\gamma$-Gap-$k$-NCP$_{p}$} instance with vector number $n$ and length $\gamma kl+n$.
\end{theorem}
\begin{proof}
    Given the instance of \textsc{$\gamma$-Gap-$k$-MLD$_{p}$} with $n$ vectors, each of length $l$, we mark them as $V=\{v_1,\cdots, v_n\}$, and let the target vector be $t$. We construct a new instance with vector set $V'=\{v_1',\cdots, v_n'\}$ and target $t'$ as follows.
    \begin{itemize}
        \item For each $i\in[n]$, $v_i'=(v_i)^{\lceil \gamma k\rceil}\circ (0^{i-1}10^{n-i})$.
        \item $t'=t^{\lceil \gamma k\rceil}\circ 0^n$.
    \end{itemize}

    We now show completeness and soundness of this reduction.
    
        \noindent\textbf{Completeness.}
        Assume there exists $k_0\leq k$ vectors $v_{i_1},\cdots, v_{i_{k_0}} \in V$ and their corresponding coefficients $c_{i_1}, \cdots, c_{i_{k_0}}$ that $\Sigma_{j\in[k_0]} c_{i,j}v_{i,j}=t$. We show that $||\Sigma_{j\in[k_0]} c_{i,j}v_{i,j}'-t'||_0 =k_0\leq k$. This is trivial since $\Sigma_{j\in[k_0]} c_{i,j}v_{i,j}'-t'=t^{\lceil \gamma k\rceil}\circ (\Sigma_{i\in[k_0]}\vec{e}_{i})-t'=0^{\lceil\gamma kl\rceil}\circ (\Sigma_{i\in[k_0]}\vec{e}_{i})$.

        \noindent\textbf{Soundness.}
        Assume that any set of vectors that spans a vector space containing $t$ must have cardinality at least $\gamma k$. Fix any set of vectors $S'=\{v'_{i_1},\cdots, v'_{i_|S|}\}\subseteq V'$, each vector $v'_{i_j}$ in $S'$ is associated with a non-zero coefficient $c_{j}$, if $|S'|\geq \gamma k$, then
        \[
        ||\Sigma_{j\in [|S'|]} c_{j} v'_{i_j} - t||_0 \geq ||\Sigma_{j\in [|S'|]} c_{j}\vec{e}_{i_{j}} - 0^{n}||_0 =\gamma k.
        \]
        If $|S'|<\gamma k$, then $\Sigma_{j\in [|S'|]} c_j v_{i_j}\neq t$, and we have
        \[
        ||\Sigma_{j\in [|S'|]} c_{j} v'_{i_j} - t||_0 \geq \lceil \gamma k \rceil||\Sigma_{j\in [|S'|]} c_j v_{i_j}-t||_0\geq \lceil \gamma k \rceil\geq \gamma k.
        \]
\end{proof}

Now we show a reduction from \textsc{$\gamma$-Gap-$k$-NCP$_{p}$} to \textsc{$\gamma$-Gap-$k$-MLD$_{p}$}. Without loss of generality, we can assume the vector length $m$ is no less than vector number $n$ in \textsc{NCP} instance.
\begin{theorem}
    There is a deterministic polynomial time reduction that on input a \textsc{$\gamma$-Gap-$k$-NCP$_p$} instance with vector number $n$ and length $m$, output a \textsc{$\gamma$-Gap-$k$-MLD$_p$} instance with vector number $m$ and length $m-n$.
\end{theorem}
\begin{proof}
    Without loss of generality assume the input instance contains $n$ linear independent vectors $\Vec{u}_1,\cdots, \Vec{u}_n$ and a target vector $\Vec{t}_0$ from $\mathbb{F}_p^m$. Let $\vec{e}_i\in \mathbb{F}_p^m$ be the unit vector having $1$ in the $i$-th entry, then for all $i\in[n]$, there exists $\vec{v}_i\in\mathbb{F}_p^{m-n}$ such that $\vec{e}_i \circ \vec{v}_i$ is the linear combination of $\Vec{u}_1,\cdots, \Vec{u}_n$. Also, there exists $\vec{t}'\in\mathbb{F}_p^{m-n}$ satisfies $(0^{n}\circ \vec{t'})-\vec{t}$ is the linear combination of $\{\vec{e}_i \circ \vec{v}_i\}_{1\leq i\leq n}$ (and also $\Vec{u}_1,\cdots, \Vec{u}_n$). 

    We show that this transformation preserves the completeness and soundness. Let $a_1,\cdots, a_n\in\mathbb{F}_p$ and $\vec{w}\in \mathbb{F}_p^m$ satisfy that $\Sigma_{1\leq i\leq n}a_i\vec{u}_i=\vec{t}+\vec{w}$. Then, there must exists $c_1,\cdots, c_n\in\mathbb{F}_p$ that $\Sigma_{1\leq i\leq n}c_i(\vec{e}_i \circ \vec{v}_i)=(0^{n}\circ \vec{t'})+\vec{w}$ since $\{\Vec{u}_1,\cdots, \Vec{u}_n\}$ and $\{\vec{e}_i \circ \vec{v}_i\}_{1\leq i\leq n}$ can linearly represent each other. Therefore, the new instance is equivalent to the original instance.

    Now consider an \textsc{$k$-MLD$_p$} instance having vector set $\{\vec{v}_1,\cdots, \vec{v}_n\}\cup \{\vec{e}_1,\cdots, \vec{e}_{m-n}\}$ (here each $\vec{e}_i$ is in $\mathbb{F}_p^{m-n}$) and target vector $\vec{t}'$. The completeness comes from that there exists $c_1,\cdots, c_n\in\mathbb{F}_p$ satisfying
    \begin{align*}
        ||\Sigma_{1\leq i\leq n}c_i(\vec{e}_i \circ \vec{v}_i) - (0^{n}\circ \vec{t'})||_0\leq k
    \end{align*}
    which means
    \begin{align*}
        ||\Sigma_{1\leq i\leq n}c_i\vec{e}_i - 0^{n}||_0+ ||\Sigma_{1\leq i\leq n}c_i\vec{v}_i - \vec{t'}||_0\leq k
    \end{align*}
    indicating that there are at most $k$ of $c_1,\cdots, c_n$ are not zero, and $||\Sigma_{1\leq i\leq n}c_i\vec{v}_i - \vec{t'}||_0\leq k-||(c_1,\cdots, c_n)||_0$. To simplify the notation, we assume that the non-zero coefficients are $c_1,\cdots, c_{k_0}$ where $k_0\leq k$. Then, in the \textsc{$k$-MLD$_p$} instance we choose $c_1,\cdots, c_{k_0}$ to be the coefficients of $\vec{v}_1,\cdots, \vec{v}_{k_0}$, and choose $(k-k_0)$ vectors in $\{\vec{e}_1, \cdots, \vec{e}_{m-n}\}$ that corresponds to $(k-k_0)$, each with the corresponding coefficient. The sum of them is exactly $\vec{t}'$ and the vector number we choose is exactly $k$.

    For soundness, we have $||\vec{w}||_0>\gamma k$ no matter what the coefficients $c_1,\cdots, c_n$ are. Assume we have a solution to the \textsc{$k$-MLD$_p$} instance with size $k'\leq \gamma k$, and without loss of generality assume that it is $\{c_1\vec{v}_1,\cdots, c_{k_0}\vec{v}_{k_0},d_1\vec{e}_1, \cdots, d_{k'-k_0}\vec{e}_{k'-k_0}\}$ (coefficients attached). Let $c_{k_0+1}=\cdots = c_n=0$, we immediately have $||\Sigma_{1\leq i\leq n} c_i\vec{v}_i - \vec{t}_0||_0=k'-k_0$. Then, we have
    \begin{align*}
        &||\Sigma_{1\leq i\leq n}c_i(\vec{e}_i \circ \vec{v}_i) - (0^{n}\circ \vec{t'})||_0\\
        =&||\Sigma_{1\leq i\leq n}c_i\vec{e}_i - 0^{n}||_0+ ||\Sigma_{1\leq i\leq n}c_i\vec{v}_i - \vec{t'}||_0\\
        =&k_0+(k'-k_0)\\
        =&k'\leq \gamma k,
    \end{align*}
    a contradiction. Hence, we have showed that the new instance is a \textsc{$\gamma$-Gap-$k$-MLD$_p$} instance. Finally, the reduction clearly runs in polynomial time and preserves the parameter $k$.
\end{proof}

\end{document}